\documentclass[11pt]{article}

\usepackage[ruled,vlined,linesnumbered]{algorithm2e}
\usepackage{amsmath,amssymb,amsthm}
\usepackage{mathtools}
\usepackage{cite}

\usepackage{enumitem}
\usepackage{graphicx}
\usepackage{ifthen}
\usepackage{subfigure}
\usepackage{array,tabularx}
\usepackage{booktabs}
\usepackage{microtype}
\usepackage[margin=1in]{geometry}

\newcommand{\newparentheses}[3]{%
  \expandafter\newcommand\csname #1\endcsname[1]{#2##1#3}%
  \expandafter\newcommand\csname #1L\endcsname[1]{\bigl#2##1\bigr#3}%
  \expandafter\newcommand\csname #1XL\endcsname[1]{\Bigl#2##1\Bigr#3}%
  \expandafter\newcommand\csname #1V\endcsname[1]{\left#2##1\right#3}}

\newparentheses{parens}{(}{)}
\newparentheses{braces}{\{}{\}}
\newparentheses{brackets}{[}{]}
\newparentheses{floor}{\lfloor}{\rfloor}
\newparentheses{ceil}{\lceil}{\rceil}
\newparentheses{abs}{|}{|}
\newparentheses{set}{\{}{\}}
\newparentheses{size}{|}{|}

\makeatletter
\newcommand{\onenewattribute}[3]{%
  \@ifundefined{#1}{\let\@@def\newcommand}{\let\@@def\renewcommand}%
  \expandafter\@@def\csname #1\endcsname[2][]{%
    \ifthenelse{\equal{##1}{}}%
    {#2\csname #3\endcsname{##2}}%
    {#2_{##1}\csname #3\endcsname{##2}}}}
\newcommand{\newattribute}[2]{%
  \onenewattribute{#1}{#2}{parens}%
  \onenewattribute{#1L}{#2}{parensL}%
  \onenewattribute{#1XL}{#2}{parensXL}%
  \onenewattribute{#1V}{#2}{parensV}}
\makeatother

\newattribute{Oh}{\mathrm{O}}
\newattribute{Thehta}{\Theta}
\newattribute{Ohmega}{\Omega}
\newattribute{oh}{\mathrm{o}}
\newattribute{ohmega}{\omega}
\newattribute{poly}{\mathrm{poly}}

\newattribute{alg}{\textsc{Maf}}
\newattribute{aalg}{\textsc{Maaf}}
\newattribute{balg}{\textsc{Refine}}

\newattribute{interval}{I}

\newcommand{\subtree}[2][]{%
  \ifthenelse{\equal{#1}{}}%
  {T(#2)}%
  {#1(#2)}}
\newcommand{\induced}[2][]{%
  \ifthenelse{\equal{#1}{}}%
  {T|#2}%
  {#1|#2}}
\newcommand{\clade}[2][]{%
  \ifthenelse{\equal{#1}{}}%
  {T_{#2}}%
  {#1_{#2}}}
\newcommand{\Clade}[3][]{%
  \ifthenelse{\equal{#1}{}}%
  {T_{#2}^{#3}}%
  {#1_{#2}^{#3}}}

\newattribute{maf}{m}
\newattribute{maaf}{\tilde{m}}
\newattribute{ecut}{e}
\newattribute{aecut}{\tilde{e}}
\newcommand{\reach}[1][]{%
  \ifthenelse{\equal{#1}{}}%
  {\sim}%
  {\sim_{#1}}}
\newcommand{\noreach}[1][]{%
  \ifthenelse{\equal{#1}{}}%
  {\nsim}%
  {\nsim_{#1}}}
\newcommand{\lca}[2][]{%
  \ifthenelse{\equal{#1}{}}%
  {LCA(#2)}%
  {LCA_{#1}(#2)}}
\newcommand{\edge}[1]{e_{#1}}
\newcommand{\triple}[3][]{%
  \ifthenelse{\equal{#1}{}}%
  {#2|#3}%
  {#1|#2|#3}}
\newcommand{\quartet}[2]{#1|#2}
\newcommand{\parent}[1]{p_{#1}}

\newcommand{\dest}[1]{d_{#1}}

\newattribute{dspr}{d_{\textit{SPR}}}
\newattribute{dmspr}{d_{\textit{sSPR}}}
\newattribute{drspr}{d_{\textit{SPR}}}
\newattribute{dtbr}{d_{\textit{TBR}}}
\newattribute{dhyb}{\textit{hyb}}

\newattribute{indeg}{\mathrm{deg}_{\mathrm{in}}}
\newattribute{hybdeg}{\mathrm{deg}_{\mathrm{in}}^-}

\newcommand{\I}{\mathcal{I}}

\newcommand{\nil}{\textrm{nil}}

\def\comment#1{}
\def\withcomments{%
  \addtolength{\oddsidemargin}{-0.5in}%
  \addtolength{\evensidemargin}{-0.5in}%
  \setlength{\marginparwidth}{1.25in}
  \newcounter{mycommentcounter}%
  \def\comment##1{\refstepcounter{mycommentcounter}%
    \ifhmode
      \unskip
      {\dimen1=\baselineskip
        \divide\dimen1 by 2%
        \raise\dimen1\llap{\tiny -\themycommentcounter-}}%
    \fi
    \marginpar{\renewcommand{\baselinestretch}{0.8}%
      \footnotesize [\themycommentcounter]: \raggedright ##1}}%
  \date{\framebox{Draft of \today}}}

\newtheorem{observation}{Observation}{\bfseries}{\itshape}
\newtheorem{lemma}{Lemma}{\bfseries}{\itshape}
\newtheorem{theorem}{Theorem}{\bfseries}{\itshape}
{\bfseries}{\itshape}

\begin{document}


\title{Fixed-Parameter and Approximation Algorithms for Maximum Agreement
  Forests of Multifurcating Trees}

\author{Chris~Whidden, Robert~G.~Beiko,	and~Norbert~Zeh%
  \thanks{
    Faculty of Computer Science, Dalhousie University, Halifax, Nova Scotia,
    Canada.\protect\\
    E-mail: $\{$whidden,beiko,nzeh$\}$@cs.dal.ca}%
  }

\markboth{IEEE/ACM Transactions on Computational Biology and Bioinformatics,
  Vol.\ X, No.\ Y,  Month Year}{Whidden \MakeLowercase{\textit{et al.}}:
  Fixed-Parameter and Approximation Algorithms for Maximum Agreement Forests}

\maketitle


\begin{abstract}
  We present efficient algorithms for computing a maximum agreement
  forest (MAF) of a pair of multifurcating (nonbinary) rooted trees.
  Our algorithms match the running times of the currently best algorithms for
  the binary case.
  The size of an MAF corresponds to the subtree prune-and-regraft (SPR) distance
  of the two trees and is intimately connected to their hybridization number.
  These distance measures are essential tools for understanding reticulate
  evolution, such as lateral gene transfer, recombination, and hybridization.
  Multifurcating trees arise naturally as a result of statistical uncertainty in
  current tree construction methods.
\end{abstract}




\section{Introduction}

Phylogenetic trees are the standard model for representing the evolution of a
set of species (taxa) through ``vertical'' inheritance~\cite{hillis96}.
Yet, genetic material can also be shared between contemporary organisms via
lateral gene transfer, recombination or hybridization.
These processes allow species to rapidly adapt to new environments as shown by,
for example, the rapid spread of antibiotic resistance and other harmful traits
in pathogenic bacteria~\cite{rosasmagallanes2006}.
Untangling vertical and lateral evolutionary histories is thus both difficult
and of great importance.
To do so often requires the comparison of phylogenetic trees for individual gene
histories with a reference tree.
Distance measures that model reticulation events using subtree prune-and-regraft
(SPR)~\cite{hillis96} and hybridization~\cite{baroni05} operations
are of particular interest in such comparisons due to their direct evolutionary
interpretations~\cite{baroni05, beiko2006pil}.

These distance measures are biologically meaningful but also NP-hard to
compute~\cite{bordewich05,hickey2008sdc, bordewich07}.
As a result, there has been significant effort to develop efficient
fixed-parameter~\cite{whidden2009uva, whidden2010fast, albrecht2012fast,
  chen2010hybridnet, 2011arXiv1108.2664W} and
approximation~\cite{rodrigues2007maf, bordewich08, whidden2009uva} algorithms to
compute these distances, most of which use the equivalent notion of maximum
agreement forests (MAFs)~\cite{hein96, bordewich05, baroni05}.
Efficient algorithms for computing these distances
have generally been restricted to binary trees.
The exceptions are reduction rules for computing hybridization numbers
of nonbinary trees~\cite{linz09hnt} and a recent depth-bounded search
algorithm~\cite{vaniersel} for computing the subtree prune-and-regraft distance
of such trees.

Multifurcations (or \emph{polytomies}) are vertices of a tree with two or more
children.
A multifurcation is \emph{hard} if it indeed represents an inferred common
ancestor which produced three or more species as direct descendants; it is
\emph{soft} if it simply represents ambiguous evolutionary
relationships~\cite{maddison1989reconstructing}.
Simultaneous speciation events are assumed to be rare, so a common assumption
is that all multifurcations are soft.
If we force the resolution of multifurcating trees into binary trees, then we
infer evolutionary relationships that are not supported by the original data and
may infer meaningless reticulation events.
Thus, it is crucial to develop efficient algorithms to compare multifurcating
trees directly.

In this paper, we extend the fastest approximation and fixed-parameter
algorithms for computing MAFs of binary rooted trees to multifurcating trees
(thus showing that computing MAFs for multifurcating trees is fixed-parameter
tractable).
The size of an MAF of two binary trees is equivalent to their SPR distance.
In keeping with the assumption that multifurcations are soft, we define an MAF
of two multifurcating trees so that its size is equivalent to what we call the
\emph{soft} SPR distance: the minimum number of SPR operations required to
transform a binary resolution of one tree into a binary resolution of the other.
This distinction avoids the inference of meaningless differences between the
trees that arises, for example, when one tree has a set of resolved bifurcations
that are part of a multifurcation in the other tree.
This is similar to the extension of the hybridization number to
multifurcating trees by Linz and Semple~\cite{linz09hnt}.

\begin{figure*}[t]
  \hspace*{\stretch{1}}%
  \subfigure[\unskip\label{fig:x-tree}]{\includegraphics{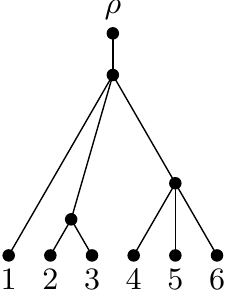}}%
  \hspace*{\stretch{2}}%
  \subfigure[\unskip\label{fig:subtree}]{\includegraphics{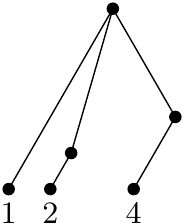}}%
  \hspace*{\stretch{2}}%
  \subfigure[\unskip\label{fig:induced}]{\includegraphics{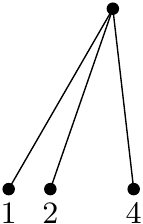}}%
  \hspace*{\stretch{2}}%
  \subfigure[\unskip\label{fig:x-tree-binary}]{\includegraphics{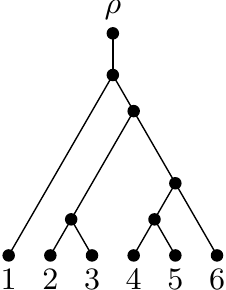}}%
  \hspace*{\stretch{1}}%
  \\
  \hspace*{\stretch{1}}%
  \subfigure[\unskip\label{fig:spr}]{\includegraphics{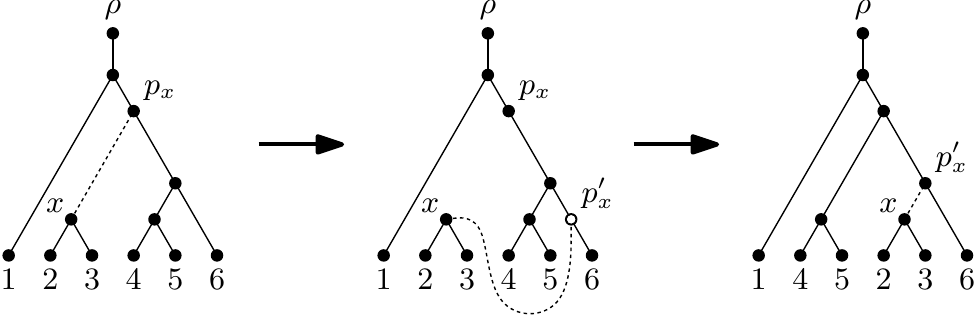}}%
  \hspace*{\stretch{1}}%
  \caption{(a)~An $X$-tree $T$.
    (b)~The subtree $\subtree{V}$ for $V = \set{1, 2, 4}$.
    (c)~$\induced{V}$.
    (d) A binary resolution of $T$.
    (e)~Illustration of an SPR operation applied to the binary resolution of
    $T$.}
\end{figure*}

Our fixed-parameter algorithm achieves the same running time as in the binary
case.
Our approximation algorithm achieves the same approximation factor as in
the binary case at the cost of increasing the running time from linear to
$\Oh{n \lg n}$.
These results are not trivial extensions of the algorithm for binary trees.
They require new structural insights and a novel method for
terminating search branches of the depth-bounded search tree,
coupled with a careful analysis of the resulting recurrence relation.

The rest of this paper is organized as follows.
Section~\ref{sec:prelim} introduces the necessary terminology and notation.
Section~\ref{sec:structure} presents the key structural results for
multifurcating agreement forests.
Section~\ref{sec:fpt} presents our new algorithms based on these results.
Finally, in Section~\ref{sec:concl}, we present closing remarks and discuss open
problems and possible extensions of this work.


\section{Preliminaries}

\label{sec:prelim}

Throughout this paper, we mostly use the definitions and notation
from~\cite{allen01, bonet06, bordewich05, bordewich08, rodrigues2007maf,
  linz09hnt}.
A \emph{(rooted phylogenetic) $X$-tree} is a rooted tree $T$ whose leaves are
the elements of a label set $X$ and whose non-root internal nodes have at least
two children each; see Figure~\ref{fig:x-tree}.
$T$~is \emph{binary} (or \emph{bifurcating}) if all internal nodes have exactly
two children each, otherwise it is \emph{multifurcating}.
The root of $T$ has label $\rho$ and has one child.
Throughout this paper, we consider $\rho$ to be a member of $X$.
For a subset $V$ of~$X$, $\subtree{V}$ is the smallest subtree of $T$ that
connects all nodes in~$V$; see Figure~\ref{fig:subtree}.
The \emph{\mbox{$V$-tree} induced by $T$} is the smallest tree $\induced{V}$
that can be obtained from $\subtree{V}$ by \emph{contracting} unlabelled nodes
with only one child, that is, by merging each such node with one of its
neighbours and removing the edge between them.
See Figure~\ref{fig:induced}.
An \emph{expansion} does the opposite:
It splits a node $v$ into two nodes $v_1$ and $v_2$ such that $v_1$ is $v_2$'s
parent and divides the children of $v$ into two subsets that become the children
of $v_1$ and~$v_2$, respectively.
For brevity, we refer to this operation as expanding the subset of $v$'s
children that become $v_2$'s children.

Let $T_1$ and $T_2$ be two $X$-trees.
We say that $T_2$ \emph{resolves} $T_1$ or, equivalently, $T_2$ is a
\emph{resolution} of $T_1$ if $T_1$ can be obtained from $T_2$ by contracting
internal edges.
$T_2$ is a \emph{binary resolution} of $T_1$ if $T_2$ is binary.
See Figure~\ref{fig:x-tree-binary}.

A \emph{subtree prune-and-regraft} (SPR) operation on a binary rooted $X$-tree
$T$ cuts an edge $x\parent{x}$, where $\parent{x}$ denotes the
parent of $x$.
This divides $T$ into subtrees $T_x$ and $T_{\parent{x}}$ containing $x$
and~$\parent{x}$, respectively.
Then it introduces a node $\parent{x}'$ into $T_{\parent{x}}$ by subdividing an
edge of $T_{\parent{x}}$ and adds an edge~$x\parent{x}'$, thereby making $x$ a
child of~$\parent{x}'$.
Finally, $\parent{x}$~is removed using a contraction.
See Figure~\ref{fig:spr}.
On a multifurcating tree, an SPR operation may also use any existing node of
$T_{\parent{x}}$ as $\parent{x}'$ and contracts $\parent{x}$ only if it
has only one child besides $x$.

\begin{figure*}[t]
  \hspace*{\stretch{1}}%
  \subfigure[\unskip\label{fig:spr-hard}]{\includegraphics{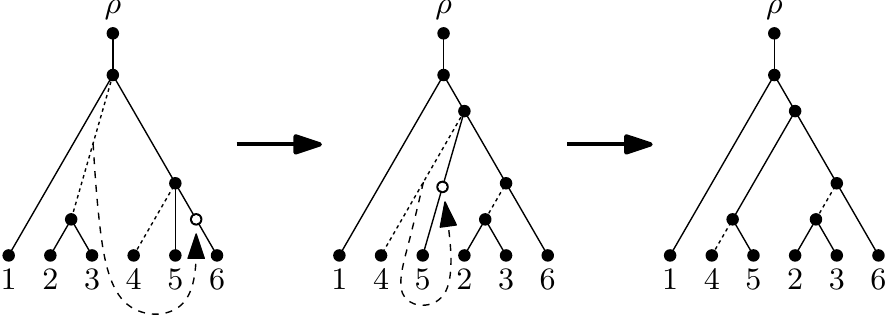}}%
  \hspace*{\stretch{2}}%
  \subfigure[\unskip\label{fig:maf}]{\includegraphics{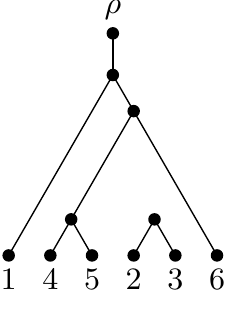}}
  \hspace*{\stretch{1}}%
  \caption{(a) SPR operations transforming the tree $T$ from
    Figure~\ref{fig:x-tree} into the second tree in Figure~\ref{fig:spr}.
    Each operation changes the top endpoint of one of the dotted edges.
    The hard SPR distance between the two trees is 2.
    (b) The MAF representing only the first transfer of (a) (equivalently,
    Figure~\ref{fig:spr}).
    The second transfer is unnecessary if the multifurcation represents
    ambiguous data rather than simultaneous speciation.
    Thus, the soft SPR distance is 1.}
  \label{fig:hard}
\end{figure*}

SPR operations give rise to a distance measure $\dspr{\cdot, \cdot}$ between
binary $X$-trees, defined as the minimum number of such operations required to
transform one tree into the other.
The trees in Figure~\ref{fig:spr}, for example, have SPR distance
$\dspr{T_1, T_2} = 1$.
An analogous distance measure, which we call the \emph{hard SPR distance},
could be defined for multifurcating \mbox{$X$-trees}; however, under the
assumption that most multifurcations are soft, this would capture differences
between the trees that are meaningless.
Instead, we define the \emph{soft SPR distance} $\dmspr{T_1,T_2}$
between two multifurcating trees $T_1$ and $T_2$ to be the minimum SPR distance
of all pairs of binary resolutions of $T_1$ and~$T_2$.\footnote{This is similar
  to the generalization of the hybridization number used by Linz and
  Semple~\cite{linz09hnt}}
For simplicity, we simply refer to this as the SPR distance in the remainder of
this paper.
These two distance measures are illustrated in Figure~\ref{fig:hard}.
Note that the soft SPR distance is not a metric but captures the minimum number
of SPR operations needed to explain the difference between the two trees.

These distance measures are related to the sizes of appropriately defined
agreement forests.
To define these, we first introduce some terminology.
For a forest $F$ whose components $T_1, T_2, \dots, T_k$ have label sets
$X_1, X_2, \dots, X_k$, we say $F$ \emph{yields} the forest with components
$\induced[T_1]{X_1}, \induced[T_2]{X_2}, \dots, \induced[T_k]{X_k}$;
if $X_i = \emptyset$, then $\subtree[T_i]{X_i} = \emptyset$ and, hence,
$\induced[T_i]{X_i} = \emptyset$.
For a subset $E$ of edges of $F$, we use $F - E$ to denote the forest obtained
by deleting the edges in $E$ from $F$, and $F \div E$ to denote the forest
yielded by $F - E$.
Thus, $F \div E$ is the contracted form of $F - E$.
We say $F \div E$ is a \emph{forest of~$F$}.

Given $X$-trees $T_1$ and $T_2$ and forests $F_1$ of $T_1$ and $F_2$ of~$T_2$, a
forest $F$ is an \emph{agreement forest} (AF) of $F_1$ and $F_2$ if it is a
forest of a binary resolution of $F_1$ and of a binary resolution of $F_2$.
$F$~is a \emph{maximum agreement forest} (MAF) of $F_1$ and $F_2$ if there is
no AF of $F_1$ and $F_2$ with fewer components.
An MAF of the trees from Figure~\ref{fig:spr-hard} is shown in
Figure~\ref{fig:maf}.
We denote the number of components in an MAF of $F_1$ and~$F_2$ by
$\maf{F_1, F_2}$, and the size of the smallest edge set $E$ such that
$F' \div E$ is an AF of $F_1$ and $F_2$ by $\ecut{F_1, F_2, F}$, where $F$ is a
forest of $F_2$ and $F'$ is a binary resolution of $F$.
Bordewich and Semple~\cite{bordewich05} showed that, for two \emph{binary}
rooted $X$-trees $T_1$ and $T_2$, $\dspr{T_1, T_2} = \ecut{T_1, T_2, T_2} =
\maf{T_1, T_2} - 1$.
This implies that
$\dmspr{T_1, T_2} = \ecut{T_1, T_2, T_2} = \maf{T_1, T_2} - 1$ for two
arbitrary rooted $X$-trees because $\dmspr{T_1, T_2}$,
$\ecut{T_1, T_2, T_2}$, and $\maf{T_1, T_2}$ are taken as the minimum over
all binary resolutions of $T_1$ and $T_2$.
Thus, to determine the SPR distance between two rooted $X$-trees, we need to
compute a binary MAF of the two trees.

We write $a \reach[F] b$ when there exists a path between two nodes $a$ and $b$
of a forest $F$.
For a node $x$ of $F$, $F^x$~denotes the subtree of $F$ induced
by all descendants of $x$, inclusive.
For two rooted forests $F_1$ and $F_2$ and a node $a \in F_1$,
we say that $a$ \emph{exists} in $F_2$ if there exists a node $a'$ in $F_2$
such that $F_1^a = F_2^{a'}$.
For simplicity, we refer to both $a'$ simply as $a$.
For forests $F_1$ and $F_2$ and nodes $a, c \in F_1$ with a common parent, we
say $\set{a, c}$ is a \emph{sibling pair} of $F_1$ if $a$ and $c$ exist in
$F_2$.
Figure~\ref{fig:sibling-pair} shows such a sibling pair.
We say $\set{a_1, a_2, \ldots, a_m}$ is a \emph{sibling group} if
$\set{a_i, a_j}$ is a sibling pair of $F_1$, for all $1 \le i < j \le m$,
and $a_1$ has no sibling not in the group.

The correctness proofs of our algorithms in the next sections make use of the
following three lemmas.
Lemma~\ref{lem:edge-shift} was shown by Bordewich et al.~\cite{bordewich08} for
binary trees.
The proof trivially extends to multifurcating trees.

\begin{lemma}
  \label{lem:edge-shift}
  Let $F$ be a forest of an \mbox{$X$-tree}, $e$ and $f$ edges of~$F$, and $E$ a
  subset of edges of $F$ such that $f \in E$ and $e \notin E$.
  Let $v_f$ be the end vertex of $f$ closest to~$e$, and $v_e$ an end vertex
  of~$e$.
  If (1) $v_f \reach[F-E] v_e$ and (2)~$x \noreach[F - (E \cup \set{e})] v_f$,
  for all $x \in X$,
  then $F \div E = F \div (E \setminus \set{f} \cup \set{e})$.
\end{lemma}

Let $F_1$ and $F_2$ be forests of $X$-trees $T_1$ and $T_2$, respectively.
Any agreement forest of $F_1$ and $F_2$ is clearly also an agreement forest of
$T_1$ and $T_2$.
Conversely, an agreement forest of $T_1$ and $T_2$ is an agreement forest of
$F_1$ and $F_2$ if it is a forest of $F_2$ and there are no two leaves $a$ and
$b$ such that $a \reach[F_2] b$ but $a \noreach[F_1] b$.
This is formalized in the following lemma.
Our algorithms ensure that any intermediate forests $F_1$ and $F_2$ they produce
have this latter property.
Thus, this lemma allows us to reason about agreement forests of $F_1$ and $F_2$
and of $T_1$ and $T_2$ interchangeably, as long as they are forests of $F_2$.

\begin{lemma}
  \label{lem:forest_ecut}
  Let $F_1$ and $F_2$ be forests of $X$-trees $T_1$ and $T_2$, respectively.
  Let $F_1$ be the union of trees $\dot{T}_1, \dot{T}_2, \ldots, \dot{T}_k$ and
  $F_2$ be the union of forests $\dot{F}_1, \dot{F}_2, \ldots, \dot{F}_k$ such
  that $\dot{T}_i$ and $\dot{F}_i$ have the same label set,
  for all $1 \le i \le k$.
  Let $F_2'$ be a resolution of $F_2$.
  $F_2' \div E$ is an AF of $T_1$ and $T_2$ if and only if it is an AF of $F_1$
  and $F_2$.
\end{lemma}

To use Lemma~\ref{lem:edge-shift} to prove structural properties of agreement
forests, which are defined in terms of resolutions of forests, we also need the
following lemma, which specifies when an expansion does not change the SPR
distance.
Its proof is provided in Section~\ref{sec:lem:expand:proof} in the supplementary
material.

\begin{lemma}
  \label{lem:expand}
  Let $F_1$ and $F_2$ be resolutions of forests of rooted $X$-trees $T_1$ and
  $T_2$, and let $F \div E$ be a maximum agreement forest of $F_1$ and $F_2$,
  where $F$ is a binary resolution of $F_2$.
  Let $a_1, a_2, \ldots, a_p, a_{p+1}, \ldots, a_m$ be the children of a node in
  $F_2$ and let $F_2'$ be the result of expanding
  $\set{a_{p+1}, a_{p+2}, \ldots, a_m}$ in $F_2$.
  If $a_i' \noreach[F \div E] a_j'$, for all $1 \le i \le p$, $p+1 \le j \le m$,
  and all leaves $a_i' \in F_2^{a_i}$ and $a_j' \in F_2^{a_j}$, then
  $\ecut{F_1, F_2, F_2} = \ecut{F_1, F_2, F_2'}$.
\end{lemma}

\begin{figure}[tb]
  \centering
  \includegraphics{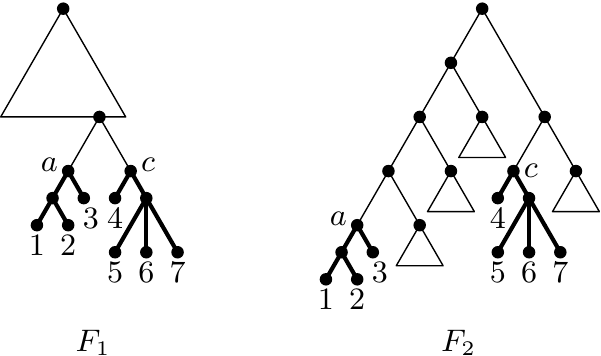}
  \caption{A sibling pair $\set{a, c}$ of two forests $F_1$ and $F_2$:
    $a$~and $c$ have a common parent in $F_1$, and both subtrees $F_1^a$ and
    $F_1^c$ exist also in $F_2$.}
  \label{fig:sibling-pair}
\end{figure}

A \emph{triple} $\triple{ab}{c}$ of a rooted forest $F$ is defined by a set
$\set{a, b, c}$ of three leaves in the same component of $F$ and such that the
path from $a$ to $b$ in $F$ is disjoint from the path from $c$ to the root of
the component.
Multifurcating trees also allow for triples $\triple[a]{b}{c}$ where
$a$, $b$, and $c$ share the same lowest common ancestor (LCA).
A triple $\triple{ab}{c}$ of a forest $F_1$ is \emph{compatible} with a forest
$F_2$ if it is also a triple of $F_2$ or $F_2$ contains the triple
$\triple[a]{b}{c}$; otherwise it is \emph{incompatible} with $F_2$.

An agreement forest of two forests $F_1$ and $F_2$ cannot contain a triple
incompatible with either of the two forests.
Thus, we have the following observation.

\begin{observation}
Let $F_1$ and $F_2$ be forests of rooted $X$-trees $T_1$ and $T_2$,
    and let $F$ be an agreement forest of $F_1$ and $F_2$.
    If $\triple{ab}{c}$ is a
    triple of $F_1$ incompatible with $F_2$, then $a \noreach[F] b$ or
    $a \noreach[F] c$.
    \label{obs:incompatible-triple}
\end{observation}

For two forests $F_1$ and $F_2$ with the same label set, two components $C_1$
and $C_2$ of $F_2$ are said to \emph{overlap} in $F_1$ if there exist leaves
$a, b \in C_1$ and $c, d \in C_2$ such that the paths from $a$ to $b$ and from
$c$ to $d$ in $F_1$ exist and are not edge-disjoint.
The following lemma is an easy extension of a lemma of \cite{bordewich08}, which
states the same result for binary trees instead of binary forests.

\begin{lemma}
  \label{lem:forest-condition}
  Let $F_1$ and $F_2$ be binary resolutions of forests of two \mbox{$X$-trees}
  $T_1$ and $T_2$, and denote the label sets of the components of $F_1$ by
  $X_1, X_2, \dots, X_k$ and the label sets of the components of $F_2$ by
  $Y_1, Y_2, \dots, Y_l$.
  $F_2$~is a forest of $F_1$ if and only if (1) for every $Y_j$, there exists an
  $X_i$ such that $Y_j \subseteq X_i$, (2) no two components of $F_2$ overlap in
  $F_1$, and (3) no triple of $F_2$ is incompatible with $F_1$.
\end{lemma}


\section{The Structure of Multifurcating\\Agreement Forests}

\label{sec:structure}

This section presents the structural results that provide the intuition and
formal basis for the algorithms presented in Section~\ref{sec:fpt}.
All these algorithms start with a pair of trees ($T_1$, $T_2$) and then cut
edges, expand sets of nodes, remove agreeing components from consideration, and
merge sibling pairs until the resulting forests are identical.
The intermediate state is that $T_1$ and $T_2$ have been resolved and reduced to
forests $F_1$ and $F_2$, respectively.
$F_1$ consists of a tree $\dot{T}_1$ and a set of components $F_0$ that exist in
$F_2$.
$F_2$ has two sets of components.
One is $F_0$.
The other, $\dot{F}_2$, has the same label set as $\dot{T}_1$ but may not agree
with $\dot{T}_1$.
The key in each iteration is deciding which edges in $\dot{F}_2$ to cut next or
which nodes to expand, in order to make progress towards an MAF of $T_1$
and~$T_2$.
The results in this section identify small edge sets in $\dot{F}_2$ such that at
least one edge in each of these sets has the property that cutting it reduces
$\ecut{T_1, T_2, F_2}$ by one.
Some of these edges are introduced by expanding nodes.
The approximation algorithm cuts all edges in
the identified set, and the size of the set gives the approximation ratio of the
algorithm.
The FPT algorithm tries each edge in the set in turn, so that the size of the
set gives the branching factor for a depth-bounded search algorithm.

\begin{figure*}[t]
  \centering
  \includegraphics{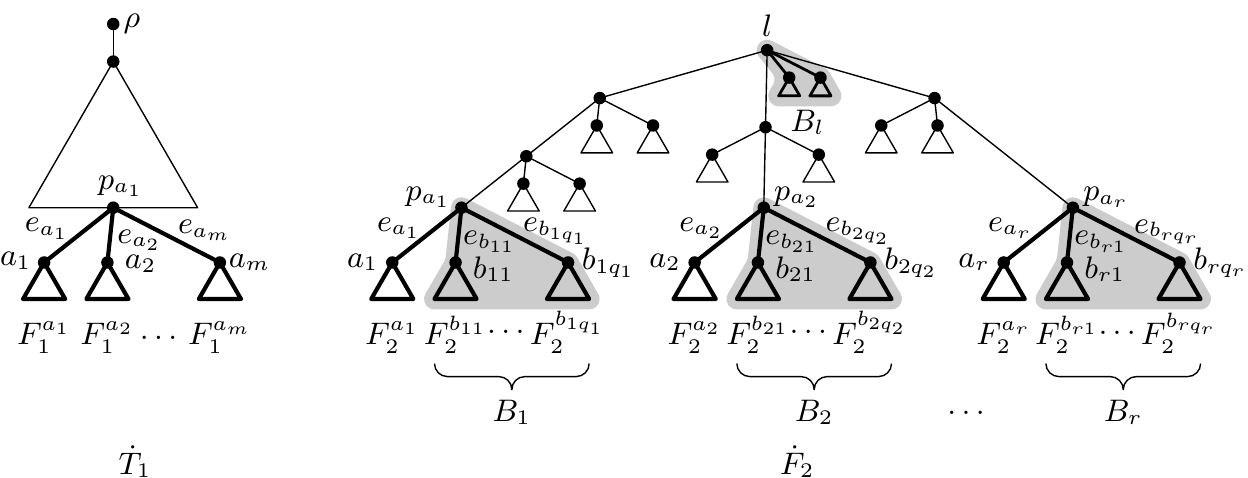}
  \caption{Tree labels for a sibling group $\set{a_1, a_2, \ldots, a_m}$
    such that $a_1, a_2, \ldots, a_r$ share a minimal LCA $l$.}
  \label{fig:non-sibling-a-definition}
\end{figure*}

Let $\set{a_1, a_2, \ldots, a_m}$ be a sibling group of $\dot{T}_1$.
If there exist indices $i \ne j$ such that $a_i$ and $a_j$ are also siblings
in~$F_2$, we can expand this sibling pair $\set{a_i, a_j}$ and replace $a_i$ and
$a_j$ with their parent node $(a_i, a_j)$ in the sibling group.
If there exists an index $i$ such that $F_2^{a_i}$ is a component of~$F_2$,
then we can cut $a_i$'s parent edge in~$F_1$, thereby removing $a_i$ from the
sibling group.
Thus, we can assume $a_i$ and $a_j$ are not siblings in~$F_2$, for all
$1 \le i < j \le m$, and $F_2^{a_i}$ is not a component of $F_2$,
for all $1 \le i \le m$.
We have $a_i \in \dot{F}_2$, for all $1 \le i \le m$, because $\dot{T}_1$ and
$\dot{F}_2$ have the same label set.
Let $B_i = \set{b_{i1}, b_{i2}, \ldots b_{iq_i}}$ be
the siblings of $a_i$ in $F_2$, for $1 \le i \le m$.
We use $\edge{x}$ to denote the edge connecting a node $x$ to its
parent~$\parent{x}$, $\edge{B_i}$ to denote the edge introduced by
expanding $B_i$, and $\parent{B_i}$ to denote the common parent of the nodes
in~$B_i$.
$F_2 - \set{\edge{B_i}}$ denotes the forest obtained from $F_2$ by expanding
$B_i$ and then cutting $\edge{B_i}$, and we use $F_2^{B_i}$ to denote the
subforest of $F_2$ comprised of the subtrees
$F_2^{b_{i1}}, F_2^{b_{i2}}, \dots, F_2^{b_{iq_i}}$.

Consider a subset $\set{a_{i_1}, a_{i_2}, \ldots, a_{i_r}}$ of a sibling group
$\set{a_1, a_2, \dots, a_m}$.
We say $a_{i_1}, a_{i_2}, \ldots, a_{i_r}$ \emph{share their LCA} $l$ if
$l = \lca[F_2]{a_i,a_j}$, for all $i,j \in \set{i_1, i_2, \ldots, i_r}$,
$i \ne j$.
If, in addition, $\lca[F_2]{a_i,a_j}$ is not a proper descendant of $l$, for all
$1 \le i < j \le m$, we say that $a_{i_1}, a_{i_2}, \ldots,
a_{i_r}$ \emph{share a minimal LCA} $l$.
For simplicity, we always order the elements of the group so that
$\set{a_{i_1}, a_{i_2}, \dots, a_{i_r}} = \set{a_1, a_2, \ldots, a_r}$ and
assume the subset that shares $l$ is maximal, that is, $a_i$ is not a descendant
of $l$, for all $r < i \le m$.
We use $B_l$ to denote the set of children of $l$ that do not have any member
$a_i$ of the sibling group as a descendant.
Note that $B_l \subseteq B_i$ when $a_i$ is a child of~$l$.
These labels are illustrated
in Figure~\ref{fig:non-sibling-a-definition}.

Our first result shows that at least one of the edges
$\edge{a_1}$, $\edge{a_2}$, $\edge{B_1}$, and $\edge{B_2}$ has the property
that cutting it reduces $\ecut{T_1, T_2, F_2}$ by one.
This implies that cutting $\edge{a_1}$, $\edge{a_2}$,
$\edge{\parent{a_1}}$, and
$\edge{\parent{a_2}}$ reduces $\ecut{T_1, T_2, F_2}$ by at least 1.

\begin{theorem}
  \label{thm:spr:fourway}
  Let $F_1$ and $F_2$ be forests of rooted $X$-trees $T_1$ and $T_2$,
  respectively, and assume $F_1$ consists of a tree $\dot{T}_1$ and a set of
  components that exist in $F_2$.
  Let $\set{a_1, a_2, \ldots, a_m}$ be a sibling group of $\dot T_1$ such that
  either $a_1, a_2, \ldots, a_r$ share a minimal LCA $l$ in $F_2$ or
  $a_2, a_3, \ldots, a_r$ share a minimal LCA $l$ in $F_2$ and
  $a_1 \noreach[F_2] a_i$, for all $2 \le i \le m$;
  $a_1$ is not a child of~$l$; $a_2$ is not a child of $l$ unless $r=2$;
  $a_i$ and $a_j$ are not siblings in~$F_2$, for all $1 \le i < j \le m$;
  and $F_2^{a_i}$ is not a component of~$F_2$, for all $1 \le i \le m$.
  Then
  \begin{enumerate}[label=(\roman*)]
  \item $\ecut{T_1, T_2, F_2 - \set{\edge{x}}} = \ecut{T_1, T_2, F_2} - 1$, for
    some $x \in \set{a_1, a_2, B_1, B_2}$.
  \item $\ecut{T_1, T_2, F_2 - \set{\edge{a_1}, \edge{a_2}, \edge{p_{a_1}},
	\edge{p_{a_2}}}} \le \ecut{T_1, T_2, F_2} - 1$.
  \end{enumerate}
\end{theorem}

\begin{proof}
  (ii) follows immediately from (i) because cutting
  $\set{\edge{a_1}, \edge{a_2}, \edge{p_{a_1}}, \edge{p_{a_2}}}$ is
  equivalent to cutting
  $\set{\edge{a_1}, \edge{a_2}, \edge{B_1}, \edge{B_2}}$.
  For (i), it suffices to prove that there exist a binary resolution
  $F$ of $F_2$ and an edge set $E$ of size
  $\ecut{T_1, T_2, F_2}$ such that
  $F \div E$ is an MAF of $T_1$ and
  $T_2$ and $E \cap \set{\edge{a_1}, \edge{a_2}, \edge{B_1}, \edge{B_2}}
  \ne \emptyset$.

  So assume $F \div E$ is an MAF of $T_1$ and $T_2$ and
  $E \cap \set{\edge{a_1}, \edge{a_2}, \edge{B_1}, \edge{B_2}}
  = \emptyset$.
  By Lemma~\ref{lem:forest_ecut}, $F \div E$ is also an MAF of $F_1$ and $F_2$.
  We prove that we can replace some edge $f \in E$ with an edge in
  $\set{\edge{a_1}, \edge{a_2}, \edge{B_1}, \edge{B_2}}$
  without changing $F \div E$.

  First assume $b_1' \noreach[F \div E] x$, for all leaves $b_1'
  \in F_2^{B_1}$ and $x \notin\nobreak F_2^{B_1}$.
  By Lemma~\ref{lem:expand}, expanding $B_1$
  does not change $\ecut{F_1, F_2, F_2}$, so we can assume $F$
  contains this expansion.%
  \footnote{In fact, using the same ideas as in the proof of
    Lemma~\ref{lem:expand}, it is not difficult to see that this expansion
    never precludes obtaining the same forest $F \div E$ by cutting a different
    set of $\size{E}$ edges.
    We discuss the importance of this
    to hybridization and reticulate analysis in Section~\ref{sec:concl}.}
  Now we choose an arbitrary leaf $b_1' \in B_1$ and
  the first edge $f \in E$ on the path from $\parent{B_1}$ to~$b_1'$.
  By Lemma~\ref{lem:edge-shift}, $F \div E = F \div (E
  \setminus \set{f} \cup \set{\edge{B_1}})$.
  If $b_2' \noreach[F \div E] x$, for all leaves $b_2'
  \in\nobreak F_2^{B_2}$ and $x \notin F_2^{B_2}$, $a_1' \noreach[F \div E]
  \parent{a_1}$, for all leaves $a_1' \in F_2^{a_1}$, or $a_2'
  \noreach[F \div E] \parent{a_2}$, for all leaves $a_2' \in F_2^{a_2}$, then
  the same argument shows that $F \div E = F \div (E \setminus \set{f} \cup
  \set{\edge{x}})$, for $x=B_2$, $x=a_1$, and $x =\nobreak a_2$, respectively.
  Thus, we can assume there exist leaves \mbox{$a_1' \in F_2^{a_1}$}, $a_2' \in
  F_2^{a_2}$, $b_1' \in F_2^{B_1}$, and $b_2' \in F_2^{B_2}$ such that
  $a_1' \reach[F \div E] \parent{a_1} \reach[F \div E] b_1'$ and
  $a_2' \reach[F \div E] \parent{a_2} \reach[F \div E] b_2'$.

  Now recall that either $a_1, a_2, \ldots, a_r$ share the minimal LCA $l$
  and $a_1$ is not a child of $l$ or
  $a_2, a_3, \ldots, a_r$ share the minimal LCA $l$ and $a_1 \noreach[F_2] a_i$,
  for all $2 \le i \le m$.
  In either case, $a_i \notin F_2^{B_1}$, for all $1 \le i \le m$.
  Since $a_i$ also is not an ancestor of $p_{B_1}$, this shows that
  $b_1' \notin F_2^{a_i}$, for all $1 \le i \le m$.
  Thus, $F_1$ contains the triple $a_1'a_2'|b_1'$,
  while $F_2$ contains the triple $a_1'b_1'|a_2'$ or $a_1' \noreach[F_2] a_2'$.
  By Observation~\ref{obs:incompatible-triple}, this
  implies that $a_1' \reach[F \div E] \parent{a_1} \reach[F \div E] b_1'
  \noreach[F \div E] a_2' \reach[F \div E] \parent{a_2} \reach[F \div E] b_2'$
  and, hence, $b_2' \notin F_2^{a_1}$.
  Since $a_2$ is not a child of $l$ unless $r = 2$, we also have
  $b_2' \notin F_2^{a_i}$, for all $2 \le i \le m$.
  Thus, $F_1$ also contains the triple $a_1'a_2'|b_2'$, which implies that
  the components of $F \div E$ containing $a_1', b_1'$ and $a_2', b_2'$
  overlap in $F_1$, a contradiction.
\end{proof}

Theorem~\ref{thm:spr:fourway} covers every case where
some minimal LCA $l$ exists.
If there is no such minimal LCA, then
each $a_i$ must be in a separate component of $F_2$.
In the following
lemma we show the stronger result that cutting $\edge{a_1}$ or $\edge{a_2}$
reduces $\ecut{T_1, T_2, F_2}$ by one in this case (which immediately
implies that claims (i) and (ii) of Theorem~\ref{thm:spr:fourway} also hold
in this case).

\begin{lemma}[Isolated Siblings]
  \label{lem:spr:sc}
  If $a_1 \noreach[F_2] a_i$, for all $i \ne 1$,
  $a_2 \noreach[F_2] a_j$, for all $j \ne 2$, and
  $F_2^{a_i}$ is not a component of $F_2$, for all $1 \le i \le m$,
  then there exist a resolution $F$ of $F_2$ and an edge set $E$ of size
  $\ecut{T_1, T_2, F_2}$ such that
  $F \div E$ is an AF of $T_1$ and $T_2$ and $E \cap
  \set{\edge{a_1}, \edge{a_2}} \ne \emptyset$.
\end{lemma}

\begin{proof}
  Consider an edge set $E$ of size $\ecut{T_1, T_2, F_2}$ and such
  that $F \div E$ is an AF of $F_1$ and $F_2$, and assume $E$
  is chosen so that $\size{E \cap \set{\edge{a_1}, \edge{a_2}}}$ is
  maximized.
  Assume for the sake of contradiction that $E \cap \set{\edge{a_1}, \edge{a_2}}
  = \emptyset$.
  Then, by the same arguments as in the proof of Theorem~\ref{thm:spr:fourway},
  there exist leaves $a_1' \in F_2^{a_1}$ and $a_2' \in F_2^{a_2}$ such that
  $a_1' \reach[F \div E] a_1$ and $a_2' \reach[F \div E] a_2$.
  Since $\set{a_1, a_2, \ldots, a_m}$ is a
  sibling group of $F_1$ but $a_1 \noreach[F_2] a_i$, for all $i \ne 1$, and
  $a_2 \noreach[F_2] a_j$, for all $j \ne 2$,
  we must have $a_1' \noreach[F - E] x$, for all
  leaves $x \notin F_2^{a_1}$, or $a_2' \noreach[F - E] x$,
  for all leaves $x \notin F_2^{a_2}$.
  W.l.o.g.\ assume the former.
  Since $F_2^{a_1}$ is not a component of~$F_2$, there exists a leaf $y
  \notin F_2^{a_1}$ such that $a_1 \reach[F_2] y$ and, hence,
  $a_1' \reach[F_2] y$.
  For each such leaf $y$, the path from $a_1'$ to $y$ in $F$ contains an edge
  in $E$ because $a_1' \noreach[F \div E] y$, and this edge does not belong to
  $F_2^{a_1}$ because $a_1' \reach[F \div E] a_1$.
  We pick an arbitrary such leaf $y$, and let $f$
  be the first edge in $E$ on the path from $a_1'$ to $y$.
  The edges $\edge{a_1}$ and $f$ satisfy the conditions of
  Lemma~\ref{lem:edge-shift}, that is, $F \div E = F \div (E
  \setminus \set{f} \cup \set{\edge{a_1}})$.
  This contradicts the choice of $E$.
\end{proof}

Theorem~\ref{thm:spr:fourway} and Lemma~\ref{lem:spr:sc} are all
that is needed to obtain a linear-time 4-approximation algorithm and
an FPT algorithm with running time $\Oh{4^k n}$ for computing rooted MAFs,
an observation made independently in~\cite{vaniersel}.
To improve on this in our algorithms in Section~\ref{sec:fpt}, we exploit
a useful observation from the proof of Theorem~\ref{thm:spr:fourway}:
if there exists an MAF $F \div E$ and leaves $a_1' \in F_2^{a_1}$ and
$b_1' \in F_2^{B_1}$ such that $a_1' \reach[F \div E] b_1'$, then $a_2'
\noreach[F \div E] b_2'$, for all $a_2' \in F_2^{a_2}$ and $b_2' \in F_2^{B_2}$.
This implies that, if we choose to cut $\edge{a_2}$
or $\edge{B_2}$ and keep both $\edge{a_1}$ and $\edge{B_1}$ in a branch of
our FPT algorithm, then we need only decide which edge, $\edge{a_j}$ or
$\edge{B_j}$, to cut in each pair $\set{\edge{a_j}, \edge{B_j}}$, for
$3 \le j \le r$, in subsequent steps of this branch.
This allows us to follow each 4-way branch in the algorithm (where we decide
whether to cut $\edge{a_1}$, $\edge{B_1}$, $\edge{a_2}$ or $\edge{B_2}$)
by a series of 2-way branches.
We cannot use this idea when a sibling group consists of
only two nodes.
The following lemma addresses this case.
Its proof is provided in Section~\ref{sec:lem:spr:twosiblings:proof}
of the supplementary material.

\begin{lemma}
  \label{lem:spr:twosiblings}
  Let $T_1$ and $T_2$ be rooted $X$-trees, and let $F_1$ be a forest of
  $T_1$ and $F_2$ a forest of $T_2$.
  Suppose $F_1$ consists of a tree
  $\dot{T}_1$ and a set of components that exist in $F_2$.
  Let $\set{a_1, a_2}$ be a sibling group of $\dot T_1$ such that neither
  $F_2^{a_1}$ nor $F_2^{a_2}$ is a component of $F_2$ and, if $a_1$ and
  $a_2$ share a minimal LCA $l$, then $a_1$ is not a child of $l$.
  In particular, $a_1$ and $a_2$ are not siblings in~$F_2$.
  Then
  \begin{enumerate}[label=(\roman*)]
  \item $\ecut{T_1, T_2, F_2 - \set{\edge{x}}} = \ecut{T_1, T_2, F_2} - 1$, for
    some $x \in \set{a_1, a_2, B_1}$.
  \item $\ecut{T_1, T_2, F_2 - \set{\edge{a_1}, \edge{a_2}, \edge{p_{a_1}}
      }} \le \ecut{T_1, T_2, F_2} - 1$.
  \end{enumerate}
\end{lemma}

Next we examine the structure of a sibling group more closely as a basis
for a refined analysis that leads to our final FPT algorithm with
running time $\Oh{2.42^k n}$.
First we require the notion of \emph{pendant subtrees} that
we will be able to cut in unison.
Let $a_1, a_2, \ldots, a_r$ be the members
of a sibling group $\set{a_1, a_2, \dots, a_m}$ that share
a minimal LCA $l$ in $F_2$, and consider the path from $a_i$ to~$l$,
for any $1 \le i \le r$.
Let $x_1, x_2, \dots, x_{s_i}$ be the nodes on this path, excluding $a_i$
and $l$.
For each $x_j$, let $B_{ij}$ be the set of children of $x_j$, excluding the
child that is an ancestor of~$a_i$, and let $F_2^{B_{ij}}$ be the subforest
of $F_2$ consisting of all subtrees $F_2^b$, $b \in B_{ij}$.
Note that $B_{i1} = B_i$, and $F_2^{B_{i1}} = F_2^{B_i}$, if $s_i > 0$.
Analogously to the definition of $\edge{B_i}$,
we use~$\edge{B_{ij}}$, for $1 \le j \le s_i$, to
denote the edge introduced by expanding the nodes in $B_{ij}$ in $F_2$ and
$F_2 - \set{\edge{B_{ij}}}$ to denote the forest obtained by expanding $B_{ij}$
and cutting edge~$\edge{B_{ij}}$.
Note that expanding $B_{ij}$ turns the forest $F_2^{B_{ij}}$ into a single
\emph{pendant subtree} attached to $x_j$.
We distinguish five cases for the structure of the subtree of $F_2$ induced
by the paths between $a_1, a_2, \dots, a_r$ and $l$:
\begin{description}
\item[Isolated Siblings:] $a_1 \noreach[F_2] a_i$, for all $i \ne 1$, and
  $a_2 \noreach[F_2] a_j$, for all $j \ne 2$.
\item[At Most One Pendant Subtree:] $a_1, a_2, \ldots, a_r$ share a
  minimal LCA $l$ in $F_2$, $s_i = 1$,  for all $1 \le i < r$, and
  $a_r$ is a child of $l$.
\item[One Pendant Subtree:] $a_1, a_2, \ldots, a_r$ share a
  minimal LCA $l$ in $F_2$ and $s_i = 1$,  for all $1 \le i \le r$.
\item[\boldmath Multiple Pendant Subtrees, $m=2$:]
  $a_1$ and $a_2$ share a minimal LCA $l$ in $F_2$ and $s_1 + s_2 \ge 2$.
\item[\boldmath Multiple Pendant Subtrees, $m>2$:]
  $a_1, a_2, \ldots, a_r$ share a minimal LCA $l$ in $F_2$ and $s_1 \ge 2$.
\end{description}

Since we assume $F_2^{a_i}$ is not a component of $F_2$, for all
$1 \le i \le m$, and no two nodes $a_i$ and $a_j$ in the sibling
group are siblings in $F_2$,
we need only consider cases where at most one $a_i$ is a
child of some minimal LCA~$l$, and we always label it $a_r$.
Hence, $s_i > 0$, for all \mbox{$i < r$} such that $a_i \reach[F_2] l$.
Thus, the five cases above cover every possible configuration of a sibling
group where we must cut an edge of $F_2$.

The following four lemmas provide stronger statements than
Theorem~\ref{thm:spr:fourway} about subsets of edges of
a resolution $F$ of $F_2$ that need to be cut in each of the last four
cases above in order to make progress towards an AF of $T_1$ and $T_2$.
All four lemmas consider a sibling group $\set{a_1, a_2, \ldots, a_m}$ of
$\dot{T}_1$ as in Theorem~\ref{thm:spr:fourway} and assume
$F_2^{a_i}$ is not a component of $F_2$, for all $1 \le i \le m$.
Lemma~\ref{lem:spr:sc} above covers the first of the five cases.

\begin{lemma}[At Most One Pendant Subtree]
  \label{lem:spr:amops}
  If $a_1, a_2, \ldots, a_r$ share a
  minimal LCA $l$ in $F_2$, $s_i=1$, for $1 \le i < r$, and $a_r$ is a child of
  $l$, then there exist a binary resolution $F$ of $F_2$ and an edge set $E$ of
  size $\ecut{T_1, T_2, F_2}$ such that $F \div E$ is an AF of $T_1$ and $T_2$
  and either $\set{\edge{B_1}, \edge{B_2}, \ldots, \edge{B_{r-1}}} \subseteq E$
  or $\set{\edge{B_1}, \edge{B_2}, \dots, \edge{B_{i-1}}, \edge{B_{i+1}},
    \edge{B_{i+2}}, \dots, \edge{B_{r-1}}} \subseteq E$ and
  $a_i' \reach[F \div E] b_i'$, for some $1 \le i \le r - 1$ and two leaves
  $a_i' \in F_2^{a_i}$ and $b_i' \in F_2^{B_i}$.
\end{lemma}

\begin{proof}
  Let $F$ be a binary resolution of $F_2$, and $E$ an edge set of size
  $\ecut{T_1, T_2, F_2}$ such that $F \div E$ is an AF of $F_1$ and $F_2$, and
  assume $F$ and $E$ are chosen so that $\size{E \cap \set{\edge{a_1},
      \edge{a_2}, \dots, \edge{a_r}, \edge{B_1}, \edge{B_2}, \dots,
      \edge{B_{r-1}}}}$ is maximized.
  Let $I := \set{i \mid 1 \le i \le r-1 \text{ and } E \cap \set{\edge{a_i},
      \edge{B_i}} \ne \emptyset}$ and
  $I' := \set{i \mid 1 \le i \le r-1 \text{ and } \edge{a_i} \in E}$.

  First observe that $\size{I} \ge r-2$.
  Otherwise there would exist two indices $1 \le i < j \le r-1$ such
  that $E \cap \set{\edge{a_i}, \edge{B_i}, \edge{a_j}, \edge{B_j}} = \nobreak
  \emptyset$.
  By the choice of $E$ and Lemma~\ref{lem:edge-shift}, this would imply that
  there exist leaves $a_i' \in F_2^{a_i}$, $b_i' \in F_2^{B_i}$, $a_j' \in
  F_2^{a_j}$, and $b_j' \in F_2^{B_j}$ such that $a_i' \reach[F \div E] b_i'$
  and $a_j' \reach[F \div E] b_j'$.
  If $a_i' \reach[F \div E] a_j'$, then $\triple{a_i'b_i'}{a_j'}$ would be a
  triple of $F \div E$ incompatible with $F_1$.
  If $a_i' \noreach[F \div E] a_j'$, the components of $F \div E$ containing
  $a_i'$ and $a_j'$ would overlap in $F_1$.
  In both cases, $F \div E$ would not be an AF of $F_1$ and $F_2$, a
  contradiction.

  Since $\size{I} \ge r-2$ and $i \notin I$ implies that
  there are two leaves $a_i' \in F_2^{a_i}$ and $b_i' \in F_2^{B_i}$ such that
  $a_i' \reach[F \div E] b_i'$, the lemma follows if we can show that there
  exist a resolution $F'$ of $F_2$ and an edge set $E'$ of size $\size{E'} \le
  \size{E}$ such that (i) $F' \div E'$ is an AF of $F_1$ and $F_2$,
  (ii) $\set{\edge{B_i} \mid i \in I} \subseteq E'$, and
  (iii) for any $1 \le i \le r-1$, there exist leaves $a_i' \in F_2^{a_i}$
  and $b_i' \in F_2^{B_i}$ such that $a_i' \reach[F' \div E'] b_i'$ if and
  only if there exist leaves $a_i'' \in F_2^{a_i}$ and $b_i'' \in F_2^{B_i}$
  such that $a_i'' \reach[F \div E] b_i''$.

  Let $Y$ be the set of leaves in all trees $F_2^{a_i}$, $i \in I' \cup
  \set{r}$, and $F_2^{B_i}$, $i \in I'$, let $Z := X \setminus Y \cup
  \set{a_r'}$, for an arbitrary leaf $a_r' \in F_2^{a_r}$,
  let $l'$ be the LCA in $F$ of all nodes in $Y$, and let $E''$ be the set of
  edges in $E$ that belong to the paths between $l'$ and the nodes
  in $\set{a_i \mid i \in I' \cup \set{r}}$.
  Since $\edge{a_i} \in\nobreak E$, for all $i \in I'$, we have $\size{E''} \ge
  \size{I'}$.
  We construct $F'$ from $F_2$ by resolving every node set $B_i$ where
  $i \in I'$, resolving the set $\set{a_r} \cup \set{\parent{a_i} \mid
    i \in\nobreak I'}$, and resolving all remaining multifurcations so that
  $\induced[F]{Y} = \induced[F']{Y}$ and $\induced[F]{Z} = \induced[F']{Z}$.
  We define the set $E'$ to be $E' := E \setminus E'' \cup \set{\edge{B_i} \mid
    i \in I'}$ if $\size{E''} = \size{I'}$;
  otherwise $E' := E \setminus E'' \cup \set{\edge{B_i} \mid i \in I'} \cup
  \set{\edge{l''}}$, where $l''$ is the LCA in $F'$ of all nodes in~$Y$.
  It is easily verified that $F'$ and $E'$ satisfy properties (ii) and (iii)
  and that $\size{E'} \le \size{E}$.
  Thus, it remains to prove that $F' \div E'$ is an AF of $F_1$ and $F_2$.

  Any triple of $F' \div E'$ incompatible with $F_1$ has to involve exactly
  one leaf $a_i' \in F_2^{a_i}$, for some $i \in I'$, because any other triple
  exists either in $F \div E$ or in $F_1$ and, thus, is compatible with $F_1$.
  Thus, any triple $\triple{a_i'}{xy}$ or $\triple{a_i'x}{y}$ of $F' \div E'$
  incompatible with $F_1$ must satisfy $x,y \notin (F')^{l''}$ because
  $\edge{B_i} \in E'$, for all $i \in I'$.
  If $\size{E''} > \size{I'}$, no such triple exists because
  $\edge{l''} \in E'$.
  If $\size{E''} = \size{I'}$, observe that $x, y \notin (F')^{l''}$ implies
  that $\triple{a_r'}{xy}$ or $\triple{a_r'x}{y}$ is also a triple of
  $F' \div E'$ incompatible with $F_1$.
  By the construction of $F'$, this triple is also a triple of $F$, and since
  $E \setminus E' = \set{\edge{a_i} \mid i \in I'}$ and $x, y \notin F_2^{a_i}$,
  for all $i \in I'$, it is also a triple of $F \div E$ incompatible with $F_1$,
  a contradiction.

  If two components of $F' \div E'$ overlap in $F_1$, let $u, v, x, y$ be four
  leaves such that $u \reach[F' \div E'] v \noreach[F' \div E'] x
  \reach[F' \div E'] y$ and the two paths $P_{uv}$ and $P_{xy}$
  between $u$ and $v$ and between $x$ and $y$, respectively, share an edge.
  Let $Y'$ be the set of leaves in the subtrees $F_2^{a_i}$, $i \in I' \cup
  \set{r}$, and assume $P_{uv}$ has no fewer endpoints in $Y'$ than $P_{xy}$.

  If both endpoints of $P_{uv}$ are in $Y'$, the two paths cannot overlap
  because $a_1, a_2, \dots, a_r$ are siblings in $F_1$ and our choices of $E$
  and $E'$ ensure that all leaves in the same subtree $F_2^{a_i}$ belong to the
  same component of $F' \div E'$.

  If both $P_{uv}$ and $P_{xy}$ have one leaf in $Y'$, their corresponding paths
  in $F' \div E'$ include $l''$.
  Thus, $u \reach[F' \div E'] x$.

  If neither $P_{uv}$ nor $P_{xy}$ has an endpoint in $Y'$, then
  $u \reach[F \div E] v$ and $x \reach[F \div E] y$.
  If $P_{uv}$ has one endpoint in $Y'$, say $u \in Y'$, and $P_{xy}$
  does not have an endpoint in $Y'$, then $x \reach[F \div E] y$ and, by the
  same arguments we used to show that no triple of $F' \div E'$ is incompatible
  with $F_1$, there exists a leaf $a_r' \in F_2^{a_r}$ such that
  $a_r' \reach[F \div E] v$ and the path from $a_r'$ to $v$ in $F_1$ overlaps
  $P_{xy}$.
  Thus, in both cases we have two paths $P_{u'v}$, $u' \in \set{u, a_r'}$, and
  $P_{xy}$ in $F_1$ such that $u' \reach[F \div E] v$ and $x \reach[F \div E] y$
  and the two paths overlap.
  Since no two components of $F \div E$ overlap in $F_1$, these two paths belong
  to the same component of $F \div E$ and w.l.o.g.\ form the quartet
  $\quartet{u'x}{vy}$, while $u' \reach[F' \div E'] v \noreach[F' \div E'] x
  \reach[F' \div E'] y$.
  Since $E' \setminus E \subseteq \set{\edge{B_i} \mid i \in I'} \cup
  \set{\edge{l''}}$, we either have $x, y \in F_2^{B_i}$ and $u', v \notin
  F_2^{a_i} \cup F_2^{B_i}$, for some $i \in I'$, or $u', v \in (F')^{l''}$,
  $u', v \notin F_2^{B_i}$, for all $i \in I'$, and $x, y \notin (F')^{l''}$.
  In the former case, these four leaves form the quartet $\quartet{u'v}{xy}$
  in $F \div E$, a contradiction.
  In the latter case, we have $u', v \in Y'$, and we already argued that this
  case is impossible.
\end{proof}

\begin{lemma}[One Pendant Subtree]
  \label{lem:spr:ops}
  If $a_1, a_2, \ldots, a_r$ share a
  minimal LCA $l$ in $F_2$, $m>2$, and \mbox{$s_i=1$}, for all $1 \le i \le r$,
  then there exist a resolution $F$ of $F_2$ and an edge set
  $E$ of size $\ecut{T_1, T_2, F_2}$ such that
  $F \div E$ is an AF of $T_1$ and $T_2$ and either
  $\set{\edge{a_1}, \edge{a_2}, \dots, \edge{a_r}} \subseteq E$,
  $\set{\edge{B_1}, \edge{B_2}, \dots, \edge{B_r}} \subseteq E$ or
  there exists an index $1 \le i \le r$ and two leaves $a_i' \in F_2^{a_i}$
  and $b_i' \in F_2^{B_i}$ such that
  $a_i' \reach[F \div E] b_i'$ and
  either $\set{\edge{a_1}, \edge{a_2}, \dots, \edge{a_{i-1}},
    \edge{a_{i+1}}, \edge{a_{i+2}}, \dots, \edge{a_r}} \subseteq E$
  or $\set{\edge{B_1}, \edge{B_2}, \dots, \edge{B_{i-1}},
    \edge{B_{i+1}}, \edge{B_{i+2}}, \dots, \edge{B_r}} \subseteq E$.
\end{lemma}

\begin{proof}
  Let $F$ be a resolution of $F_2$, and $E$ an edge set of size
  $\ecut{T_1, T_2, F_2}$ such that $F \div E$ is an AF of $F_1$ and $F_2$,
  and assume $F$ and $E$ are chosen so that $\size{E \cap
    \set{\edge{a_1}, \edge{a_2}, \dots, \edge{a_r},
      \edge{B_1}, \edge{B_2}, \dots, \edge{B_r}}}$ is maximized.

  If there exists an index $1 \le j \le r$ such that
  $\edge{B_j} \in\nobreak E$, then assume w.l.o.g.\ that $j = r$.
  The forests $F_1$ and $F_2' := F_2 \div \set{\edge{B_r}}$ satisfy the
  conditions of Lemma~\ref{lem:spr:amops}.
  Hence, there exist a resolution $F'$ of $F_2'$ and an edge set $E'$
  of size $\ecut{T_1, T_2, F_2'} = \ecut{T_1, T_2, F_2} - 1$ such that
  $F' \div E'$ is an AF of $F_1$ and $F_2'$ and, thus, of $F_1$ and $F_2$ and
  either $\set{\edge{B_1}, \edge{B_2}, \dots, \edge{B_{r-1}}} \subseteq E'$ or
  $a_i' \reach[F' \div E'] b_i'$ and
  $\set{\edge{B_1}, \edge{B_2}, \dots, \edge{B_{i-1}}, \edge{B_{i+1}},
    \edge{B_{i+2}}, \dots, \edge{B_{r-1}}} \subseteq E'$, for some index
  $1 \le i \le r-1$ and leaves $a_i' \in F_2^{a_i}$ and $b_i' \in F_2^{B_i}$.
  Thus, the resolution $F''$ of $F_2$ such that $F'' \div \set{\edge{B_r}} =
  F'$ and the edge set $E' \cup \set{\edge{B_1}}$ satisfy the lemma.

  If $E \cap \set{\edge{B_1}, \edge{B_2}, \dots, \edge{B_r}} = \emptyset$, then
  by the same arguments as in Lemma~\ref{lem:spr:amops}, we can have at most
  one index $1 \le i \le r$ such that $a_i' \reach[F \div E] b_i'$, for two
  leaves $a_i' \in F_2^{a_i}$ and $b_i' \in F_2^{B_i}$.
  If such an index $i$ exists, then $\set{\edge{a_1}, \edge{a_2}, \dots,
    \edge{a_{i-1}}, \edge{a_{i+1}}, \edge{a_{i+2}}, \dots, \edge{a_r}}
  \subseteq E$.
  If no such index exists, then $\set{\edge{a_1}, \edge{a_2}, \dots,
    \edge{a_r}} \subseteq E$.
\end{proof}

\begin{lemma}[Multiple Pendant Subtrees, $m=2$]
  \label{lem:spr:mps2}
  If $\set{a_1, a_2}$ is a sibling group such that $a_1$ and $a_2$ share a
  minimal LCA $l$ in $F_2$ and $s_1 + s_2 \ge 2$,
  then there exist a resolution $F$ of $F_2$ and an edge set $E$ of size
  $\ecut{T_1, T_2, F_2}$ such that
  $F \div E$ is an AF of $T_1$ and $T_2$ and either $E \cap
  \set{\edge{a_1}, \edge{a_2}} \ne \emptyset$ or
  $\set{\edge{B_{11}}, \edge{B_{12}}, \ldots, \edge{B_{1s_1}}} \cup
  \set{\edge{B_{21}}, \edge{B_{22}}, \ldots, \edge{B_{2s_2}}}
  \subseteq E$.
\end{lemma}
\begin{proof}
  We prove this by induction on $s = s_1 + s_2$.
  The base case is $s = 1$\footnote{We excluded this case from the statement of
  the lemma, in order to keep the cases covered by the different lemmas
  disjoint, but the lemma also holds for $s = 1$.
  A similar comment applies to Lemma~\ref{lem:spr:mps}.}
	and the claim follows from Lemma~\ref{lem:spr:twosiblings}.

  Having established the base case, we can assume $s > 1$ and the lemma
  holds for all \mbox{$1 \le s' < s$}.
  By Theorem~\ref{thm:spr:fourway}, there exist a resolution $F$ of
  $F_2$ and an edge set $E$ of size $\ecut{T_1, T_2, F_2}$ such that $F
  \div E$ is an AF of $F_1$ and $F_2$ and
  $E \cap \set{\edge{a_1}, \edge{a_2}, \edge{B_1}, \edge{B_2}} \ne \emptyset$.
  Assume $F$ and $E$ are chosen so that $\size{E \cap
    \set{\edge{a_1}, \edge{a_2}, \edge{B_1}, \edge{B_2}}}$ is maximized.
  If $E \cap \set{\edge{a_1}, \edge{a_2}} \ne \emptyset$, the lemma holds, so
  assume the contrary.
  If $s_2 = 0$ and
  $E \cap \set{\edge{a_1}, \edge{a_2}, \edge{B_1}, \edge{B_2}} =
  \set{\edge{B_2}}$, the choice of $F$ and $E$ and Lemma~\ref{lem:edge-shift}
  imply that there exist leaves $a_1' \in F_2^{a_1}$, $b_1' \in F_2^{B_1}$,
  $a_2' \in F_2^{a_2}$, and $x \notin F_2^{a_2}$ such that
  $a_1' \reach[F_2 \div \set{\edge{B_2}}] b_1'$ and
  $a_2' \reach[F_2 \div \set{\edge{B_2}}] x$.
  Thus, $a_1$ and $a_2$ satisfy the conditions of Lemma~\ref{lem:spr:sc} after
  cutting edge~$\edge{B_2}$, which implies that
  we can choose $F$ and $E$ so that $E \cap \set{\edge{a_1}, \edge{a_2}} \ne
  \emptyset$ in addition to $\edge{B_2} \in E$, a contradiction.
  Now assume $s_2 \ne 0$ or
  $E \cap \set{\edge{a_1}, \edge{a_2}, \edge{B_1}, \edge{B_2}} \ne
  \set{\edge{B_2}}$.
  In this case, $E \cap \set{\edge{B_{11}}, \edge{B_{21}}} \ne \emptyset$.
  Assume w.l.o.g.\ that $\edge{B_{11}} \in\nobreak E$.
  Then the inductive hypothesis shows that there exist
  a resolution $F'$ of $F_2$ and an edge set $E'$ of size
  $\ecut{T_1, T_2, F_2}$ such that $F' \div E'$ is an AF of
  $F_1$ and $F_2 \div \set{\edge{B_{11}}}$ (and, hence, of $F_1$ and~$F_2$)
  such that either $E' \cap \set{\edge{a_1}, \edge{a_2}} \ne \emptyset$ or
  $\set{\edge{B_{11}}, \edge{B_{12}}, \ldots, \edge{B_{1s_1}}} \cup
  \set{\edge{B_{21}}, \edge{B_{22}}, \ldots, \edge{B_{2s_2}}}
  \subseteq E'$.
  Thus, the lemma also holds in this case.
\end{proof}

The proofs of the following two lemmas are similar to that of Lemma 9.
They are provided in Sections~\ref{sec:lem:spr:mps:proof}
and~\ref{sec:lem:spr:mps-take3:proof} of the supplementary material.

\begin{lemma}[Multiple Pendant Subtrees, $m > 2$ and $r > 2$]
  \label{lem:spr:mps}
  If $a_1, a_2, \ldots, a_r$ share a
  minimal LCA~$l$ in $F_2$, $m>2$, $r>2$, and $s_1 \ge 2$, then
  there exist a resolution $F$ of $F_2$ and an edge set $E$ of size
  $\ecut{T_1, T_2, F_2}$ such that
  $F \div E$ is an AF of $T_1$ and $T_2$ and
  $E \cap \set{\edge{a_1}, \edge{a_2}} \ne \emptyset$,
  \mbox{$\set{\edge{B_{11}}, \edge{B_{12}}, \ldots, \edge{B_{1s_1}}}
    \subseteq E$} or
  \mbox{$\set{\edge{B_{21}}, \edge{B_{22}}, \ldots, \edge{B_{2s_2}}}
    \subseteq E$}.
\end{lemma}

\begin{lemma}[Multiple Pendant Subtrees, $m > 2$ and $r = 2$]
  \label{lem:spr:mps-take3}
  If $a_1, a_2, \ldots, a_r$ share a
  minimal LCA~$l$ in $F_2$, $m>2$, $r = 2$ and $s_1 \ge 2$, then
  there exist a resolution $F$ of $F_2$ and an edge set $E$ of size
  $\ecut{T_1, T_2, F_2}$ such that
  $F \div E$ is an AF of $T_1$ and $T_2$ and
  $E \cap \set{\edge{a_1}, \edge{a_2}} \ne \emptyset$,
  $\set{\edge{B_{11}}, \edge{B_{12}}, \ldots, \edge{B_{1s_1}}}
  \subseteq E$ or
  $\set{\edge{B_{21}}, \edge{B_{22}}, \ldots, \edge{B_{2s_2}}, \edge{B_2'}}
  \subseteq E$, where $B_2'$ is the set of siblings of $a_i$ after cutting
  $\edge{B_{21}}, \edge{B_{22}}, \dots, \edge{B_{2s_2}}$.
\end{lemma}


\section{MAF Algorithms}

\label{sec:fpt}

In this section, we present an FPT algorithm for computing MAFs of
multifurcating rooted trees.
This algorithm also forms the basis for a $3$-approximation algorithm
with running time $\Oh{n \log n}$, which is presented in
Section~\ref{sec:root-maf-appr} of the supplementary material.

As is customary for FPT algorithms, we focus on the decision version of the
problem: ``Given two rooted \mbox{$X$-trees} $T_1$ and $T_2$
and a parameter $k$, is $\dmspr{T_1, T_2} \le k$?''  To compute the
distance between two trees, we start with $k = 0$ and increase it
until we receive an affirmative answer.
This does not increase the running time of the algorithm by more than a constant
factor, as the running time depends exponentially on $k$.

Our FPT algorithm is recursive.
Each invocation $\alg{F_1, F_2, k, a_0}$ takes two (partially
resolved) forests $F_1$ and $F_2$ of
$T_1$ and $T_2$, a parameter~$k$, and (optionally) a node $a_0$ that
exists in $F_1$ and $F_2$ as inputs.
$F_1$ is the union of a tree $\dot{T}_1$ and a forest~$F_0$ disjoint
from~$\dot{T}_1$, while $F_2$ is the union of the same forest $F_0$ and
another forest $\dot{F}_2$ with the same label set as $\dot{T}_1$.
The output of the invocation $\alg{F_1, F_2, k, a_0}$ satisfies two
conditions: (i) If $\ecut{T_1, T_2, F_2} > k$, the output is ``no''.
(ii) If $\ecut{T_1, T_2, F_2} \le k$ and either $a_0 = \nil$ or
there exists an MAF $F$ of $F_1$ and $F_2$ such that $a_0$ is not a root of $F$
and $a_0 \noreach[F] a_i$, for every sibling $a_i$ of $a_0$ in~$F_1$,
the output is ``yes''.
Since the top-level invocation is $\alg{T_1, T_2, k, \nil}$, these two
conditions ensure that this invocation decides whether
$\ecut{T_1, T_2, T_2} \le k$.

The representation of the input to each recursive call includes two sets of
labelled nodes:
$R_d$~(roots-done) contains the roots of $F_0$, $R_t$
(roots-todo) contains the roots of (not necessarily maximal) subtrees that agree
between $\dot{T}_1$ and $\dot{F}_2$.
We refer to the nodes in these sets by their labels.
For the top-level invocation, $F_1 = \dot{T}_1 = T_1$, $F_2 =
\dot{F}_2 = T_2$, and $F_0 = \emptyset$; $R_d$ is empty and $R_t$
contains all leaves of $T_1$; $a_0 = \nil$.

$\alg{F_1, F_2, k, a_0}$ uses the results from Section~\ref{sec:structure}
to identify a small collection
$\set{E_1, E_2, \ldots, E_q}$ of subsets of edges of $\dot{F}_2$ such
that $\ecut{T_1, T_2, F_2} \le k$ only if $\ecut{T_1, T_2, F_2
  \div E_i} \le k - \size{E_i}$, for at least one \mbox{$1 \le i \le q$}.
It calls $\alg{F_1, F_2 \div E_i, k - \size{E_i}, a_i'}$ recursively,
for each subset~$E_i$ and an appropriate parameter $a_i'$,
and returns ``yes'' if and only if one of these recursive calls does.

A na\"{i}ve use of the structural results from Section~\ref{sec:structure}
would explore many overlapping edge subsets.
For example, one branch of the algorithm may cut an edge $\edge{a_i}$ and then
an edge $\edge{a_j}$, while a sibling branch may cut $\edge{a_j}$ and
then $\edge{a_i}$.
As we hinted at in Section~\ref{sec:structure},
if we cut edge $\edge{a_i}$ or its sibling edge $\edge{B_i}$ in two sibling
invocations, then there is no need to consider cutting either of these two
edges in their sibling invocations or their descendants.
Using the results of Lemmas~\mbox{\ref{lem:spr:amops}--\ref{lem:spr:mps-take3}},
we obtain more generally:
if we cut $\edge{a_i}$ or its set of progressive siblings
$\set{\edge{B_{i1}}, \edge{B_{i2}}, \ldots, \edge{B_{i s_i}}}$ in two
sibling invocations, then we need not consider these edge sets in their
sibling invocations or their descendants.
Thus, we set $a_0 = a_i$ in these sibling invocations and thereby instruct
the algorithm to ignore these edges as candidates for cutting.
An invocation $\alg{F_1, F_2, k, a_0}$ with $a_0 \ne \nil$
(Step~\ref{case:tw:non-sibling} below) makes only two
recursive calls when it would make significantly more recursive calls if
$a_0 = \nil$ (Step~\ref{case:spr:non-sibling} below).
This is not a trivial change, as it is required to obtain the
running time claimed in Theorem~\ref{thm:maf:time}.
The steps of our procedure are as follows.

\begin{enumerate}[label=\arabic{*}.,ref=\arabic{*},leftmargin=*]
\item \label{case:abort} (Failure) If $k < 0$, then
  $\ecut{T_1, T_2, F_2} \ge 0 > k$.
  Return ``no'' in this case.
\item \label{case:success} (Success) If $|R_t| = 1$, then
  $F_1 = F_2$.
  Hence, $F_2$ is an AF of $T_1$ and $T_2$, that is,
  $\ecut{T_1, T_2, F_2} = 0 \le k$.
  Return ``yes'' in this case.
\item \label{case:spr:singleton} (Prune maximal agreeing subtrees)
  If there is no node $r \in R_t$ that is a root in $F_2$, proceed to
  Step~\ref{item:choose-sib-pair}.
  Otherwise choose such a node $r \in R_t$; remove it from~$R_t$ and
  add it to $R_d$, thereby moving the corresponding subtree of
  $\dot{F}_2$ to $F_0$; and cut the edge $\edge{r}$ in $F_1$.
  If $r$'s parent $\parent{r}$ in $F_1$ now has only one child, contract
  $\parent{r}$.
  If $a_0 \ne \nil$ and $\parent{r}$'s only child before the contraction
  was $a_0$, set $a_0 = \nil$.
  Note that these changes affect only~$F_1$.
  Thus, $\ecut{T_1, T_2, F_2}$ remains unchanged.
  Return to Step~\ref{case:success}.
\item \label{item:choose-sib-pair}Choose a sibling group $\set{a_1, a_2,
  \ldots, a_m}$ in $\dot{T}_1$ such that $a_1, a_2, \ldots, a_m \in R_t$.
  If two or more members of the sibling group chosen in this invocation's
  parent invocation remain in $\dot{T}_1$, choose that sibling group.
\item \label{case:spr:sibling} (Grow agreeing subtrees)
  While there exist indices $1 \le i < j \le m$ such that $a_i$ and $a_j$
  are siblings in $\dot{F}_2$, do the following: Remove $a_i$ and $a_j$ from
  $R_t$; resolve $a_i$ and $a_j$ in $\dot T_1$ and $\dot F_2$; label their new
  parent in both forests with $(a_i, a_j)$ and add it to $R_t$.
  The new node $(a_i, a_j)$ becomes a member of the current sibling group
  and $m$ decreases by 1.
  If $m = 1$ after resolving all such sibling pairs $\set{a_i, a_j}$, contract
  the parent of the only remaining member of the sibling group
  and return to Step~\ref{case:success};
  otherwise proceed to Step~\ref{item:choose-order}.
\item \label{item:choose-order} If $a_i \noreach[F_2] a_j$,
  for all $1 \le i < j \le m$, proceed to Step~\ref{case:tw:non-sibling}.
  Otherwise there exists a node $l$ that is a minimal LCA of a group of
  nodes in the current sibling group.
  If the most recent minimal LCA chosen in an ancestor invocation is a minimal
  LCA of a subset of nodes in the current sibling group, choose $l$ to be this
  node; otherwise choose $l$ arbitrarily.
  Now order the nodes in the sibling group
  $\set{a_1, a_2, \ldots, a_m}$ so that, for some $r \ge 2$, $a_1, a_2, \ldots,
  a_r$ are descendants of $l$, while, for all $1 \le i \le r < j \le m$,
  either the LCA of $a_i$ and $a_j$ is a proper ancestor of $l$ or
  $a_i \noreach[F_2] a_j$.
  Order $a_1, a_2, \ldots, a_r$ so that $s_1 \ge s_2 \ge \cdots \ge s_r$.
  (Recall that $s_i$ is the number of nodes on the path from $a_i$ to $l$,
  excluding $a_i$ and~$l$.)
  The order of $a_{r+1}, a_{r+2}, \ldots, a_m$ is arbitrary.
  \begin{figure*}[t]
    \centering
    \includegraphics{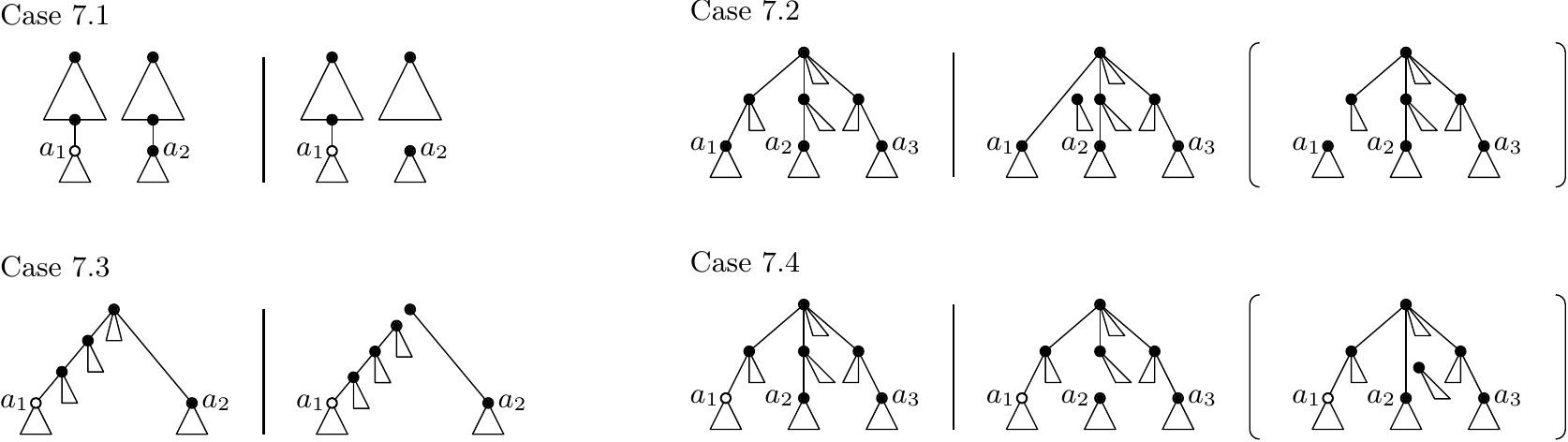}%
    \caption{The different cases of Step~\ref{case:tw:non-sibling}.
      Only the subtree of $\dot{F}_2$ rooted in $l$ is shown.
      The left side shows a possible input for each case, the right side
      visualizes the cuts made in each recursive call.
      The node $a_0$ is shown as a hollow circle (if it is a descendant of
      $l$).}
    \label{fig:tw:cases}
  \end{figure*}
\item \label{case:tw:non-sibling} (Two-way branching) If $a_0 = \nil$,
  proceed to Step~\ref{case:spr:non-sibling}.
  Otherwise distinguish four cases, where $x = 1$ if \mbox{$a_1 \ne a_0$},
  and $x = 2$ otherwise (see Figure~\ref{fig:tw:cases}).

  \begin{enumerate}[label=\theenumi.\arabic{*}.,leftmargin=*,ref=\theenumi.\arabic{*}]
  \item \label{case:tw:is} If $a_1 \noreach[F_2] a_i$, for all $i \ne 1$,
    and $a_2 \noreach[F_2] a_j$, for all $j \ne 2$,
    call $\alg{F_1, F_2 \div \set{\edge{a_x}}, k-1, a_0}$.
  \item \label{case:tw:mpsi} If $a_1, a_2, \ldots, a_r$ share the
    minimal LCA $l$ in $F_2$ and $a_0$ is not a descendant of $l$,
    call $\alg{F_1, F_2 \div \set{\edge{B_{11}},\edge{B_{12}}, \ldots,
        \edge{B_{1s_1}}}, k-s_1, a_0}$.
    If $s_1 > 1$ or $s_r >\nobreak 0$, also call
    $\alg{F_1, F_2 \div \set{\edge{a_1}}, k-1, a_0}$.
  \item \label{case:tw:mps-ladder}
    If $r = 2$, $s_x = 0$, $a_0$ is a descendant of $l$, and either $l$
    is a root of $F_2$ or its parent has a member $a_i$ of the current
    sibling group as a child, call $\alg{F_1, F_2 - \set{\edge{B_x}},
      k-1, a_0}$.
  \item \label{case:tw:mps} If $a_1, a_2, \ldots, a_r$ share the
    minimal LCA $l$ in $F_2$, $a_0$ is a descendant of $l$, and
    Case~\ref{case:tw:mps-ladder} does not apply,
    call $\alg{F_1, F_2 \div \set{\edge{a_x}}, k-1, a_0}$.
    If $m > 2$ and \mbox{$r > 2$}, make another recursive call
    $\alg{F_1, F_2 \div \set{\edge{B_{x1}},\edge{B_{x2}}, \ldots,
        \edge{B_{xs_x}}}, k-s_x, a_0}$.
    If $m > 2$ but $r = 2$, the second recursive call is
    $\alg{F_1, F_2 \div \set{\edge{B_{x1}},\edge{B_{x2}}, \ldots,
        \edge{B_{xs_x}}, \edge{B'_x}},\break k-(s_x+1), a_0}$,
    where $B_x'$ is the set of siblings of $x$ after cutting
    $\edge{B_{x1}}, \edge{B_{x2}}, \ldots, \edge{B_{xs_x}}$.
  \end{enumerate}
  Return ``yes'' if one of the recursive calls does; otherwise return
  ``no''.
\item \label{case:spr:non-sibling} (Unconstrained branching) Distinguish seven
  cases and choose the first case that applies (see
  Figure~\ref{fig:spr:cases}):
  \begin{enumerate}[label=\theenumi.\arabic{*}.,leftmargin=*,ref=\theenumi.\arabic{*}]
  \item \label{case:spr:is} If $a_1 \noreach[F_2] a_i$, for all $i \ne 1$,
    and $a_2 \noreach[F_2] a_j$, for all $j \ne 2$,
    call $\alg{F_1, F_2 \div \set{\edge{a_1}}, k-1,
      \nil}$ and $\alg{F_1, F_2 \div \set{\edge{a_2}}, k-1, \nil}$.
  \item \label{case:spr:amops} If $s_i = 1$,  for $1 \le i < r$, and
    $a_r$ is a child of $l$, call
    $\alg{F_1, F_2 \div \set{\edge{B_1},\edge{B_2},
        \ldots, \edge{B_{r-1}}},\break k-(r-1), \nil}$ and
    $\alg{F_1, F_2 \div \set{\edge{B_1}, \edge{B_2}, \ldots,\break
        \edge{B_{i-1}}, \edge{B_{i+1}}, \edge{B_{i+2}}, \ldots,
        \edge{B_{r-1}}}, k-(r-2), a_i}$, for all $1 \le i \le {r-1}$.
  \item \label{case:spr:mps2} If $m=2$, and $s_1 + s_2 \ge 2$,
    call
    $\alg{F_1, F_2 \div \set{\edge{a_1}}, k-1, \nil}$,
    $\alg{F_1, F_2 \div \set{\edge{a_2}}, k-1, \nil}$, and
    $\alg{F_1, F_2 \div \set{\edge{B_{11}},\edge{B_{12}}, \ldots,
        \edge{B_{1s_1}}} \cup \set{\edge{B_{21}},\edge{B_{22}}, \ldots,
        \edge{B_{2s_2}}}, k-(s_1 + s_2), \nil}$.
  \item \label{case:spr:ops-ladder}
    If $m > 2$, $r = 2$, and $s_1 = s_2 = 1$, call
    $\alg{F_1, F_2 \div \set{\edge{a_1}, \edge{a_2}}, k-2, \nil}$,
    $\alg{F_1, F_2 \div \set{\edge{B_1}, \edge{B_2}},\break k-2, \nil}$,
    $\alg{F_1, F_2 \div \set{\edge{B_1}, \edge{B_1'}}, k-2, a_2}$, and
    $\alg{F_1, F_2 \div \set{\edge{B_2}, \edge{B_2'}}, k-2, a_1}$, where
    $B_i'$ is the set of siblings of $a_i$ after cutting edge $B_i$.
    If $l$ is a root, $l$'s parent $\parent{l}$ has at least one
    child that is neither $l$ nor a member $a_j$ of the current sibling group
    or $l$'s grandparent has at least one child that is neither $\parent{l}$
    nor a member $a_h$ of the current sibling group, make two additional
    calls
    $\alg{F_1, F_2 \div \set{\edge{a_1}}, k-1, a_2}$ and
    $\alg{F_1, F_2 \div \set{\edge{a_2}}, k-1, a_1}$.
  \item \label{case:spr:ops} If $m > 2$, $r > 2$, and $s_i = 1$,
    for all $1 \le i \le r$, call
    $\alg{F_1, F_2 \div \set{\edge{a_1}, \edge{a_2}, \ldots, \edge{a_r}},
      k - r,\break \nil}$,
    $\alg{F_1, F_2 \div \set{\edge{B_1}, \edge{B_2}, \ldots, \edge{B_r}},
      k - r,\break \nil}$,
    $\alg{F_1, F_2 \div \set{\edge{a_1}, \edge{a_2}, \ldots, \edge{a_{i-1}},
        \edge{a_{i+1}},\break \edge{a_{i+2}}, \ldots, \edge{a_r}}, k - (r - 1),
      a_i}$, for all $1 \le i \le r$, and
    $\alg{F_1, F_2 \div \set{\edge{B_1}, \edge{B_2}, \ldots, \edge{B_{i-1}},
        \edge{B_{i+1}},\break \edge{B_{i+2}}, \ldots, \edge{B_r}}, k - (r - 1),
      a_i}$, for all $1 \le i \le r$.
  \item \label{case:spr:mps-ladder}
    If $m > 2$, $r = 2$, $s_1 \ge 2$, $s_2 = 0$ and either $l$ is a root
    of $F_2$ or its parent has a member $a_i$ of the current sibling group as
    a child, call $\alg{F_1, F_2 - \set{\edge{a_1}}, k-1, \nil}$,
    $\alg{F_1, F_2 - \set{\edge{B_{11}}, \edge{B_{12}}, \ldots,
        \edge{B_{1s_1}}}, k - s_1, \nil}$, and
    $\alg{F_1, F_2 - \set{\edge{B_2}}, k-1, \nil}$.
  \item \label{case:spr:mps} If $m>2$, $s_1 \ge 2$, and
    Case~\ref{case:spr:mps-ladder} does not apply, call
    $\alg{F_1, F_2 \div \set{\edge{a_1}}, k-1, \nil}$,
    $\alg{F_1, F_2 \div \set{\edge{a_2}}, k-1, a_1}$, and
    $\alg{F_1, F_2 \div \break \set{\edge{B_{11}},\edge{B_{12}}, \ldots,
        \edge{B_{1s_1}}}, k-s_1, \nil}$.
    If $r>2$, call
    $\alg{F_1, F_2 \div \set{\edge{B_{21}},\edge{B_{22}}, \ldots,
        \edge{B_{2s_2}}}, k-s_2,\break a_1}$.
    If $r = 2$, call
    $\alg{F_1, F_2 \div \set{\edge{B_{21}},\edge{B_{22}}, \ldots,\break
        \edge{B_{2s_2}}, \edge{B'_2}}, k-(s_2+1), a_1}$, where
    $B_2'$ is defined as in Lemma~\ref{lem:spr:mps-take3}.
  \end{enumerate}
  Return ``yes'' if one of the recursive calls does; otherwise return
  ``no''.
\end{enumerate}

\begin{theorem}
  \label{thm:maf:time}
  Given two rooted $X$-trees $T_1$ and $T_2$ and a parameter $k$, it
  takes $\Oh{(1+\sqrt{2})^k n} = \Oh{2.42^k n}$ time to decide whether
  $\ecut{T_1, T_2, T_2} \le k$.
\end{theorem}

\begin{proof}
  We use the algorithm in this section, invoking it as
  $\alg{T_1, T_2, k, \nil}$.
  We leave its correctness proof to a separate lemma
  (Lemma~\ref{lem:maf:correct} below) and focus on bounding its
  running time here.
  As we argue in Section~\ref{sec:linear-time} of the supplementary material,
  each invocation $\alg{F_1, F_2, k', a_0}$ takes $\Oh{n}$ time.
  Thus, it suffices to bound the number of invocations by
  $\Oh{(1 + \sqrt{2})^k}$.
  Let $I(k,t)$ be the number of invocations that are descendants
  of an invocation $\alg{F_1, F_2, k, a_0}$ in the recursion tree,
  where $t = 1$ if the invocation executes
  Step~\ref{case:tw:non-sibling} but not Step~\ref{case:spr:non-sibling};
  otherwise $t = 0$.
  We develop a recurrence relation for $I(k,t)$ and use it to show
  that $I(k,t) \le (1+\sqrt{2})^{2 + \max(0,k-t+3)}+2(t-1)$, which proves our
  claim.\footnote{This is a fairly loose bound on $I(k,t)$, but it is easy to
    manipulate.}

  An invocation with $t = 0$ by definition either executes neither
  Step~\ref{case:tw:non-sibling} nor Step~\ref{case:spr:non-sibling}, or it
  executes Step~\ref{case:spr:non-sibling}.
  By considering the different cases of Step~\ref{case:spr:non-sibling}, we
  obtain the following recurrence for the case when $t = 0$:
  \begin{equation*}
    I(k,0) \le
    \begin{cases}
      1                             & \text{no recursion}\\
      1 + 2I(k-1, 0)                & \text{Case~\ref{case:spr:is}}\\
      1 + I(k-1, 0) + I(k, 1)       & \text{Case~\ref{case:spr:amops}}\\
      1 + 2I(k-1, 0) + I(k-2, 0)    & \text{Cases~\ref{case:spr:mps2},
                                        \ref{case:spr:mps-ladder}}\\
      1 + 2I(k-2, 0) + 2I(k-1, 1)\\
      \phantom{1} + 2I(k-2, 1)    & \text{Case~\ref{case:spr:ops-ladder}}\\
      1 + 2I(k-3, 0) + 3I(k-2, 0)\\
      \phantom{1} + 3I(k-2, 1)      & \text{Case~\ref{case:spr:ops}}\\
      1 + I(k-1, 0) + I(k-2, 0)\\
      \phantom{1} + 2I(k-1,1)       & \text{Case~\ref{case:spr:mps}}
    \end{cases}
  \end{equation*}
  The values of the first arguments of all $I(\cdot, \cdot)$ terms are easily
  verified for Cases~\ref{case:spr:is}, \ref{case:spr:mps2},
  \ref{case:spr:ops-ladder}, and~\ref{case:spr:mps-ladder}.
  For Case~\ref{case:spr:amops}, we observe that $r \ge 2$ and the worst case
  arises when $r = 2$, giving the claimed recurrence.
  For Case~\ref{case:spr:mps}, we observe again that $r \ge 2$.
  If $r > 2$, then $s_2 > 0$, giving the recurrence for this case.
  If $r = 2$, then $s_2$ may be $0$, but the fourth recursive call cuts
  $s_2 + 1$ edges, thereby giving the same recurrence as when $r > 2$.
  For Case~\ref{case:spr:ops}, finally, we have $r \ge 3$ and, once again, the
  minimum value, $r = 3$, is the worst case, which gives the recurrence.

  \begin{figure*}[!tp]
    \centering
    \includegraphics{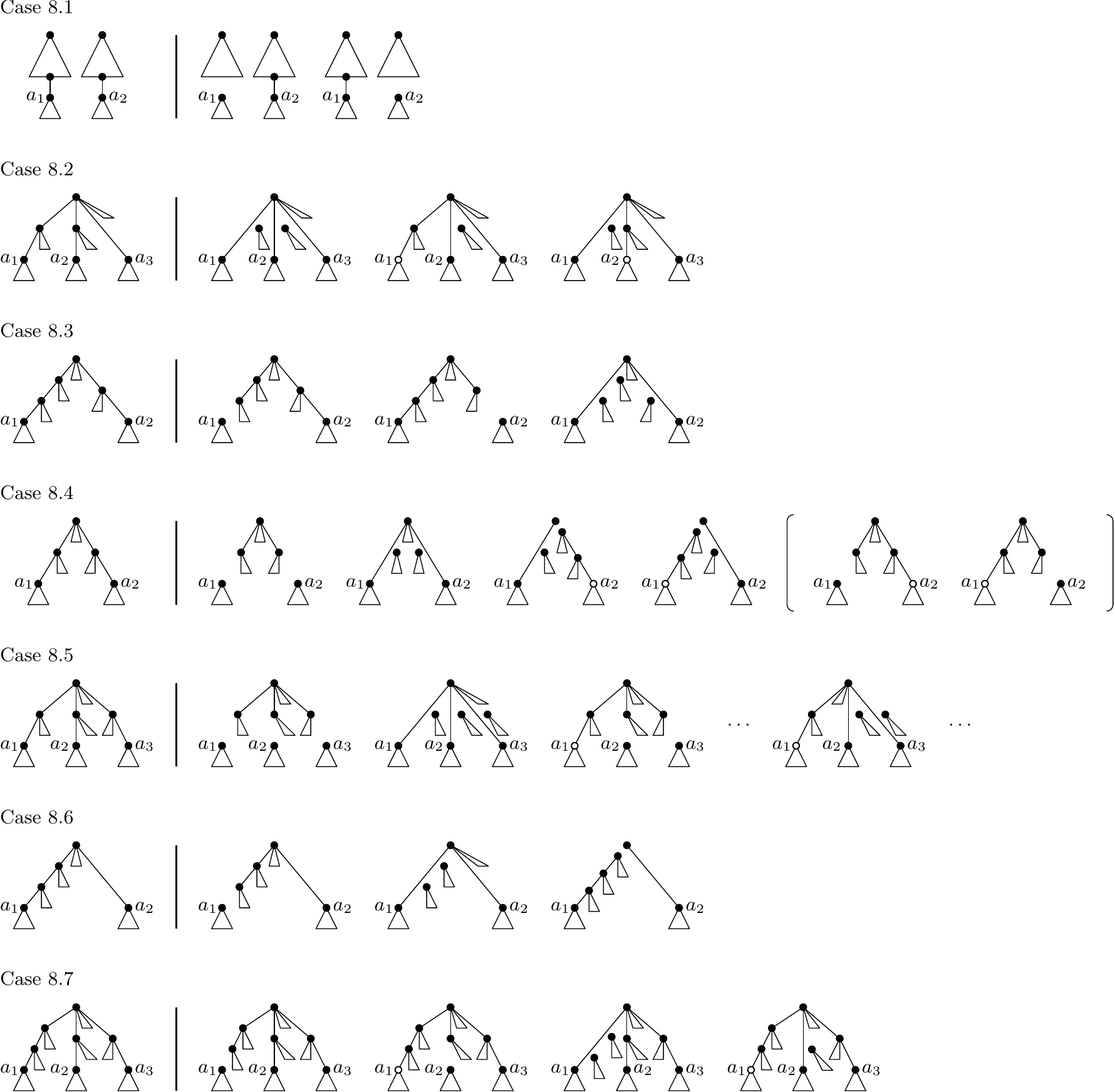}%
    \caption{The different cases of Step~\ref{case:spr:non-sibling}.
      Only the subtree of $\dot{F}_2$ rooted in $l$ is shown.
      The left side shows a possible input for each case, the right side
      visualizes the cuts made in each recursive call.
      Whenever $a_0 \ne \nil$ in a recursive call, it is shown as a hollow
      circle.
      The last two calls in Case~\ref{case:spr:ops-ladder} may or may not be
      made.}
    \label{fig:spr:cases}
  \end{figure*}

  Next we argue about the correctness of all second arguments that are $1$
  in these recurrences.
  Each such term $I(\cdot, 1)$ corresponds to a recursive call
  with $a_0 \ne\nobreak \nil$.
  Thus, in order to justify setting $t = 1$, we need to show that each such
  invocation executes Step~\ref{case:tw:non-sibling} but not
  Step~\ref{case:spr:non-sibling}, which follows if in this child invocation,
  the current invocation's sibling group exists and at least one additional
  edge cut in $F_2$ is required to makes this sibling group agree between
  $F_1$ and $F_2$---we say that the sibling group agrees between $F_1$ and $F_2$
  if $F_2$ does not contain a triple incompatible with $F_1$ and involving
  descendants of at least two members of the sibling group, and there are no two
  paths between leaves in $F_2$ that (i) belong to different components of
  $F_2$, (ii) overlap in $F_1$, and (iii)~have at least one endpoint each that
  is a descendant of a member of the current sibling group.
  The only cases that make recursive calls with $a_0 \ne \nil$ are
  Cases~\ref{case:spr:amops}, \ref{case:spr:ops-ladder}, \ref{case:spr:ops}
  and~\ref{case:spr:mps}.
  We consider each case in turn.

  In Case~\ref{case:spr:amops}, consider a recursive call that cuts edges
  $\edge{B_1}, \edge{B_2}, \ldots, \edge{B_{i-1}}, \edge{B_{i+1}},
  \edge{B_{i+2}}, \ldots, \edge{B_{r-1}}$, for some $1 \le i \le r - 1$.
  Assume w.l.o.g.\ that $i = 1$.
  After cutting edges $\edge{B_2}, \edge{B_3}, \dots, \edge{B_{r-1}}$,
  there still exist leaves $a_1' \in F_2^{a_1}$, $b_1' \in\nobreak F_2^{B_1}$,
  and $a_r' \in F_2^{a_r}$ such that $F_1$ contains the triple
  $a_1'a_r'|b_1'$ while $F_2$ contains the triple $a_1'b_1'|a_r'$.
  Thus, the sibling group does not agree between $F_1$ and $F_2$ yet.

  In Case~\ref{case:spr:ops-ladder}, if we make only four recursive calls,
  we can ignore whether setting $a_0 = a_1$ or $a_0 = a_2$ in the third and
  fourth recursive calls translates into $t = 1$ for these recursive calls
  because even the recurrence $I(k, 0) = 1 + 4I(k-2, 0)$ is bounded by the
  recurrence $I(k, 0) = 1 + 2I(k-2, 0) + 2I(k-2, 1) + 2I(k-1, 1)$ for the case
  when we make six recursive calls.
  If we make six recursive calls,
  there exists a member $a_3$ of the current sibling group that
  is not a descendant of~$l$.
  Thus, the conditions for making the fifth and sixth recursive calls imply that
  there exist leaves $a_3' \in F_2^{a_3}$ and $b_3' \notin F_2^{a_i}$, for all
  $1 \le i \le m$, such that $a_3' \reach[F_2] b_3'$ and the
  path from $a_3'$ to $b_3'$ in $F_2$ is disjoint from $F_2^l$.
  After cutting $\edge{B_1}$ and $\edge{B_1'}$,
  there exist leaves $a_2' \in\nobreak F_2^{a_2}$ and $b_2' \in F_2^{B_2}$ such
  that $a_2' \reach[F_2 \div \set{\edge{B_1}, \edge{B_1'}}] b_2'$.
  Thus, we have two paths in different connected components (between
  $a_2'$ and $b_2'$ and between $a_3'$ and $b_3'$) that overlap in $F_1$, and
  the sibling group does not agree between $F_1$ and $F_2$ yet.
  After cutting $\edge{a_1}$, we obtain overlapping paths as above if
  $a_2 \noreach[F_2] a_3$; otherwise $a_2'b_2'|a_3'$ is a triple of $F_2$
  incompatible with $F_1$.
  Thus, once again the sibling group does not agree between $F_1$
  and~$F_2$.
  Similar arguments show that setting $t = 1$ is correct when cutting edge
  $\edge{a_2}$ or edges $\edge{B_2}$ and~$\edge{B_2'}$.

  In Case~\ref{case:spr:ops}, we claim that it is correct to set $t = 1$
  for each recursive call that cuts edges
  $\edge{B_1}, \edge{B_2}, \dots, \edge{B_{i-1}},\break \edge{B_{i+1}},
  \edge{B_{i+2}}, \dots, \edge{B_r}$, for some $1 \le i \le r$.
  Assume w.l.o.g.\ that $i = 1$.
  After cutting edges $\edge{B_2}, \edge{B_3}, \dots, \edge{B_r}$,
  there exist leaves $a_1' \in F_2^{a_1}$, $b_1' \in F_2^{B_1}$, and $a_2' \in
  F_2^{a_2}$ such that $\triple{a_1'a_2'}{b_1'}$ is a triple of $F_1$ and
  $\triple{a_1'b_1'}{a_2'}$ is a triple of~$F_2$.
  Thus, the sibling group does not agree between $F_1$ and~$F_2$ yet.

  In Case~\ref{case:spr:mps}, observe that $l$ is not a root and
  its parent does not have a member of the current sibling group as a child.
  Thus, since $m > 2$, there exists an index $j$ and two leaves
  $a_j' \in F_2^{a_j}$ and $b_j' \notin F_2^{a_h}$, for all
  \mbox{$1 \le h \le m$}, such that $a_j' \reach[F_2] b_j'$ and the path from
  $a_j'$ to $b_j'$ in $F_2$ is disjoint from the path between any two leaves in
  $F_2^{a_1}$ and $F_2^{B_{11}}$.
  After cutting $\edge{a_2}$, there exist leaves $a_1' \in F_2^{a_1}$ and
  $b_1' \in F_2^{B_{11}}$ such that
  $a_1' \reach[F_2 \div \set{\edge{a_2}}] b_1'$.
  If $a_1' \reach[F_2] a_j'$, we thus have a triple $a_1'b_1'|a_j'$ of $F_2$
  incompatible with $F_1$.
  If $a_1' \noreach[F_2] a_j'$, the path between $a_1'$ and $b_1'$
  overlaps the path between $a_j'$ and $b_j'$ in $F_1$.
  In either case, the current sibling group does not agree between $F_1$
  and $F_2$ yet.
  This concludes the correctness proof of the recurrence for $I(k, 0)$.

  For $t = 1$, we distinguish whether or not the current invocation $\I$ makes a
  recursive call with $t = 0$ and whether it makes one or two recursive calls.
  If $\I$ makes no recursive call with $t = 0$,
  we obtain $I(k, 1) \le 1 + 2I(k-1, 1)$ because each case of
  Step~\ref{case:tw:non-sibling} makes at most two recursive calls, with
  parameters no greater than $k-1$.
  If $\I$ makes only one recursive call, with $t = 0$, we obtain
  $I(k, 1) \le 1 + I(k-1, 0)$ because this recursive call has parameter no
  greater than $k-1$.
  Finally, if $\I$ makes two recursive calls, at least one of them with $t = 0$,
  $\I$ must have applied Case~\ref{case:tw:mpsi} or~\ref{case:tw:mps}.
  Let $\I'$ be one of the invocations $\I$ makes with $t = 0$.
  If $t = 0$ for invocation $\I'$ because $\I'$ terminates in Step~1 or~2,
  we obtain
  $I(k,1) \le 2 + I(k-1,0)$ by counting invocations $\I$ and $\I'$ and the
  number of recursive calls spawned by the sibling invocation of $\I'$, which
  cannot be more than $I(k-1,0)$.
  So assume that $t = 0$ for invocation $\I'$ and that $\I'$ does make further
  recursive calls.
  Then the sibling group chosen in invocation $\I$ must agree between the
  input forests of invocation $\I'$.

  If invocation $\I$ applies Case~\ref{case:tw:mpsi} and makes two recursive
  calls, we observe that $m \ge 3$ because $a_0$ is a member of $\I$'s
  sibling group, $a_0$ is not a descendant of $l$, and $l$ has at least two
  descendants in the sibling group.
  Furthermore, $s_1 > 0$.
  Thus, after cutting $\edge{a_1}$, $a_2$ has a sibling forest
  $B_2'$ that does not include~$a_0$.
  Since $a_0$ also has a sibling forest $B_0$ that does not include~$a_2$,
  $\I$'s sibling group cannot agree between $\dot T_1$
  and $\dot F_2$ after cutting~$\edge{a_1}$.
  This implies that $t = 1$ for the first recursive call
  $\alg{F_1, F_2 \div \set{\edge{a_1}}, k-1, a_0}$, and $\I'$ is the second
  recursive call $\alg{F_1, F_2 \div \set{\edge{B_{11}}, \edge{B_{12}}, \ldots,
      \edge{B_{1s_1}}}, k - s_1, a_0}$.
  This gives the recurrence $I(k,1) = 1 + I(k-1,1) +\break I(k-s_1, 0)$.
  Since no two members of $\I$'s sibling group are siblings in $F_2$
  and $a_0$ is not a descendant of $l$, cutting edges $\edge{B_{11}},
  \edge{B_{12}}, \ldots, \edge{B_{1s_1}}$ can make $\I$'s sibling group
  agree between $F_1$ and $F_2$ only if $r=2$, $s_2 = 0$, and either $l$ is a
  root of $F_2$ or the only pendant nodes of the path from $l$ to the root of
  its component in $F_2$ are members of $\I$'s sibling group.
  Thus, since we assume we make two recursive calls, we must have
  $s_1 \ge 2$, that is, the recurrence for this case is
  $I(k,1) \le 1 + I(k-1,1) + I(k-2,0)$.

  Finally, if invocation $\I$ applies Case~\ref{case:tw:mps}, observe that,
  since $a_0$ is a descendant of $l$ and has a group of sibling trees $B_0$
  that do not contain any member $a_i$ of $\I$'s sibling group, this
  sibling group can be made to agree between $F_1$ and $F_2$ only by cutting
  $\edge{a_x}$.
  Moreover, since no member $a_i$ of $\I$'s sibling group is a root of
  $F_2$, cutting $\edge{a_x}$ can make this sibling group agree
  between $F_1$ and $F_2$ only if $m = 2$.
  Thus, Case~\ref{case:tw:mps} makes only one recursive call and we obtain
  $I(k,1) = 1 + I(k-1,0)$ in this case.

  By combining the different possibilities for the case when $t = 1$, we
  obtain the recurrence
  \begin{align*}
    I(k,1) \le \max(&1 + 2I(k-1, 1), 2 + I(k-1, 0),\\
    &1 + I(k-1, 1) + I(k-2, 0)).
  \end{align*}
  Simple substitution now shows that
  $I(k,t) \le (1 + \sqrt{2})^{2 + \max(0,k-t+3)} + 2(t - 1)$.
\end{proof}

\begin{lemma}
  \label{lem:maf:correct}
  For two rooted $X$-trees $T_1$ and $T_2$ and a parameter~$k$,
  the invocation $\alg{T_1, T_2, k, \nil}$ returns ``yes'' if and only
  if $\dmspr{T_1, T_2} = \ecut{T_1, T_2, T_2} \le k$.
\end{lemma}

\begin{proof}
  We use induction on $k$ to prove the following two claims, which together
  imply the lemma:
  (i)~If $\ecut{T_1, T_2, F_2} > k$, the invocation $\alg{F_1, F_2, k, a_0}$
  returns ``no''.
  (ii) If $\ecut{T_1, T_2, F_2} \le k$ and either $a_0 = \nil$ or there exists
  an MAF $F$ of $F_1$ and $F_2$ such that $a_0$ is not a root of $F$ and
  $a_0 \noreach[F] a_i$, for every sibling $a_i$ of $a_0$ in $F_1$,
  the invocation $\alg{F_1, F_2, k, a_0}$ returns ``yes''.

  (i) Assume $\ecut{T_1, T_2, F_2} > k$.
  If $k < 0$, the invocation returns ``no'' in Step~\ref{case:abort}.
  If $k \ge 0$, assume for the sake of contradiction that the invocation returns
  ``yes''.
  If it does so in Step~\ref{case:success}, then $F_2$ is an AF of $T_1$ and
  $T_2$, that is, $\ecut{T_1, T_2, F_2} = 0 \le k$, a contradiction.
  Otherwise it returns ``yes'' in Step~\ref{case:tw:non-sibling}
  or~\ref{case:spr:non-sibling}.
  Thus, there exists a child invocation $\alg{F_1', F_2', k', a_0'}$
  that returns ``yes'', where $F_2' = F_2 \div E$ and $k' = k - \size{E}$, for
  some non-empty edge set $E$.
  By the inductive hypothesis, we therefore have $\ecut{T_1, T_2, F_2'} \le k'$
  and, hence, $\ecut{T_1, T_2, F_2} \le k' + \size{E} = k$, again a
  contradiction.

  (ii) Assume $\ecut{T_1, T_2, F_2} \le k$
  and either $a_0 = \nil$ or there exists an MAF $F$ of $F_1$ and $F_2$ such
  that $a_0$ is not a root of $F$ and $a_0 \noreach[F] a_i$, for every sibling
  $a_i$ of $a_0$ in $F_1$.
  In particular, $k \ge 0$ and the invocation $\alg{F_1, F_2, k, a_0}$
  produces its answer in Step~\ref{case:success}, \ref{case:tw:non-sibling}
  or~\ref{case:spr:non-sibling}.
  If it produces its answer in Step~\ref{case:success}, it answers ``yes''.
  Next we consider Steps~\ref{case:tw:non-sibling}
  and~\ref{case:spr:non-sibling} and prove that
  at least one of the recursive calls made in each case returns ``yes'',
  which implies that the current invocation returns ``yes''.

  In Step~\ref{case:tw:non-sibling}, $a_0 \ne \nil$, that is,
  $a_0$ is not a root of $F$ and $a_0 \noreach[F] a_i$, for every sibling $a_i$
  of $a_0$ in $F_1$.
  In Case~\ref{case:tw:is}, $a_x \noreach[F_2] a_i$ and, hence, $a_x
  \noreach[F] a_i$, for all $i \ne x$.
  Thus, $a_x$~is a root of $F$ because otherwise the components of $F$
  containing $a_x$ and $a_0$ would overlap in $F_1$.
  This implies that there exists an edge set $E$ such that $\edge{a_x} \in E$
  and $F = F_2 \div E$, that is, the recursive call
  $\alg{F_1, F_2 \div \set{\edge{a_x}}, k - 1, a_0}$ returns ``yes''.

  In Case~\ref{case:tw:mpsi}, $a_1' \noreach[F] b_1'$, for all leaves
  $a_1' \in F_2^{a_1}$ and $b_1' \in F_2^{B_{1j}}$, $1 \le j \le s_1$,
  because the path between $a_1'$ and $b_1'$ would overlap the component
  containing $a_0$.
  Thus, there exists an edge set $E$ such that $F_2 \div E = F$ and either
  $\edge{a_1} \in E$ or $\set{\edge{B_{11}}, \edge{B_{12}}, \dots,
    \edge{B_{1s_1}}} \subseteq E$.
  If we make two recursive calls in Case~\ref{case:tw:mpsi}, this shows that one
  of them returns ``yes''.
  If we make only one recursive call, we have $s_1 = 1$ and $s_r = 0$.
  Assume this recursive call returns ``no''.
  Then $\edge{a_1} \in E$.
  Let $F'$ be the forest obtained by cutting $\edge{B_1}$ instead of
  $\edge{a_1}$, resolving the pair $\set{a_1, a_r}$, cutting edge
  $\edge{(a_1, a_r)}$ instead of $\edge{a_r}$ if $\edge{a_r} \in E$, and
  otherwise cutting the same edges as in $E$.
  $F$ and $F'$ are identical, except that in $F'$,
  $a_1$ and $a_r$ are siblings and no leaf in $F_2^{B_1}$ can reach a leaf
  not in $F_2^{B_1}$.
  The former cannot introduce any triples incompatible with $F_1$ because
  $a_1$ and $a_r$ are siblings in $F_1$.
  The latter cannot introduce any overlapping components because this would
  imply that $a_1' \reach[F] b_1'$, for two leaves $a_1' \in F_2^{a_1}$
  and $b_1' \in F_2^{B_1}$, and we already argued that no such path can
  exist.
  Thus, $F'$ is also an MAF of $F_1$ and $F_2$.
  Finally, $a_0$ is not a root in $F'$ and $a_0 \noreach[F'] a_1$ because
  otherwise $a_0 \reach[F] a_r$.
  Thus, since there exists an edge set $E'$ such that $F' = F_2 \div E'$ and
  $\edge{B_1} \in E'$, the recursive call $\alg{F_1, F_2 \div \set{\edge{B_1}},
    k-1, a_0}$ returns ``yes'', a contradiction.

  In Case~\ref{case:tw:mps-ladder}, there exists an edge set $E$ such that
  $F = F_2 \div E$ and $E \cap \set{\edge{a_x}, \edge{B_x}} \ne \emptyset$.
  Otherwise we would have $a_x \reach[F] a_0$ or $a_x' \reach[F]
  b_x'$, for two leaves $a_x' \in F_2^{a_x}$ and $b_x' \notin F_2^{a_i}$, for
  all $1 \le i \le m$, but we have $a_x \noreach[F] a_0$ and the path between
  $a_x'$ and $b_x'$ would overlap the component of $F$ that contains $a_0$.
  Now, if $\edge{B_x} \notin E$, we construct an MAF $F'$ of $F_1$ and
  $F_2$ such that $a_0$ is not a root in $F'$ and $a_0 \noreach[F'] a_i$, for
  all $1 \le i \le m$, and an edge set $E'$ such that $F' = F_2 \div E'$ and $\edge{B_x} \in E'$
  as in Case~\ref{case:tw:mpsi}.
  Thus, the invocation $\alg{F_1, F_2 \div \set{\edge{B_x}}, k-1, a_0}$
  returns ``yes''.

  In Case~\ref{case:tw:mps}, finally, one of the two recursive calls
  $\alg{F_1, F_2 \div \set{\edge{a_x}}, k-1, a_0}$ or
  $\alg{F_1, F_2 \div \set{\edge{B_{x1}}, \edge{B_{x2}}, \dots,
      \edge{B_{xs_x}}}, k - s_x, a_0}$ must return ``yes'',
  by the same arguments as in Case~\ref{case:tw:mpsi}.
  Thus, if $m > 2$ and $r > 2$, one of the recursive calls we make returns
  ``yes''.
  If $m > 2$ and $r = 2$, we observe that after cutting edges
  $\edge{B_{x1}}, \edge{B_{x2}}, \dots, \edge{B_{xs_x}}$, $a_x$ and $a_0$
  satisfy the conditions of Case~\ref{case:tw:mps-ladder}, which shows
  that we can cut edge $B_x'$ immediately after cutting these edges.
  Thus, if the recursive call $\alg{F_1, F_2 \div \set{\edge{a_x}}, k-1, a_0}$
  returns ``no'', the recursive call
  $\alg{F_1, F_2 \div \set{\edge{B_{x1}}, \edge{B_{x2}}, \dots,
      \edge{B_{xs_x}}, \edge{B_x'}}, k - s_x, a_0}$ must return ``yes'' in this
  case.
  Finally, if $m = 2$, since $a_x \noreach[F] a_0$ and any path between two
  leaves $a_x' \in F_2^{a_x}$ and $b_x' \notin F_2^{a_x} \cup F_2^{a_0}$ would
  overlap the component of $F$ that contains $a_0$, $a_x$ is a root in $F$.
  Thus, the invocation $\alg{F_1, F_2 \div \set{\edge{a_x}}, k-1, a_0}$ returns
  ``yes'' in this case.

  Now consider Step~\ref{case:spr:non-sibling}.
  In Cases~\ref{case:spr:is}, \ref{case:spr:amops}, \ref{case:spr:mps2},
  and~\ref{case:spr:ops}, Lemmas~\ref{lem:spr:sc}, \ref{lem:spr:amops},
  \ref{lem:spr:mps2}, \ref{lem:spr:ops} and the inductive hypothesis show that
  one of the recursive calls returns ``yes'' and, thus, the current invocation
  returns ``yes''.

  In Case~\ref{case:spr:ops-ladder}, Lemma~\ref{lem:spr:ops} and the inductive
  hypothesis show that one
  of the recursive calls
  $\alg{F_1, F_2 \div \set{\edge{a_1}, \edge{a_2}}, k-2, \nil}$,
  $\alg{F_1, F_2 \div \set{\edge{B_1}, \edge{B_2}}, k-2, \nil}$,
  $\alg{F_1, F_2 \div \set{\edge{B_1}}, k-1, a_2}$,
  $\alg{F_1, F_2 \div \set{\edge{B_2}}, k-1,\break a_1}$,
  $\alg{F_1, F_2 \div \set{\edge{a_1}}, k-1, a_2}$ or
  $\alg{F_1, F_2 \div \set{\edge{a_2}}, k-1, a_1}$
  would return ``yes''.
  Thus, we need to argue only that we can cut \emph{both} $\edge{B_i}$ and
  $\edge{B_i'}$ in the third and fourth recursive calls, and that the last
  two recursive calls are not necessary when we do not make them.

  First consider cutting $\edge{B_1}$ and setting $a_0 = a_2$ in the third
  recursive call.
  We require this call to return ``yes'' only if the other calls return ``no''.
  The forest $F_2' = F_2 \div \set{\edge{B_1}}$ contains a triple
  $\triple{a_1'}{a_2'b_2'}$, where $a_1' \in F_2^{a_1}$, $a_2' \in F_2^{a_2}$,
  and $b_2' \in F_2^{B_2}$, while $F_1$ contains the triple
  $\triple{a_1'a_2'}{b_2'}$.
  Thus, in order to obtain an MAF of $F_1$ and $F_2'$ where $a_2$ exists and
  is not a root, we need to cut either $\edge{a_1}$ or $\edge{B_1'}$.
  Since the first and fifth recursive calls return ``no'', however, we know that
  cutting $a_1$ cannot lead to an MAF of $F_1$ and $F_2$, and
  we can cut $\edge{B_1'}$ along with edge $\edge{B_1}$.
  The case when we cut $\edge{B_2'}$ along with edge $\edge{B_2}$ is analogous.

  If we do not make the recursive call $\alg{F_1, F_2 \div \set{\edge{a_1}},
    k-1, a_2}$, then $l$ has exactly one sibling, $a_i$, and its parent
  $\parent{l}$ is either a root or has exactly one sibling,~$a_j$.
  Thus, the forest $F_2' = F_2 \div \set{\edge{a_1}}$ contains a triple
  $\triple{a_2'b_2'}{a_i'}$, where $a_2' \in F_2^{a_2}$, $b_2' \in F_2^{B_2}$,
  and $a_i' \in F_2^{a_i}$, while $F_1$ contains the triple
  $\triple{a_2'a_i'}{b_2'}$.
  In order to obtain an MAF of $F_1$ and $F_2'$ where $a_2$ exists and is
  not a root, we therefore need to cut either $\edge{a_i}$ or $\edge{B_i} =
  \edge{l}$.
  If $\parent{l}$ is a root, cutting either edge has the same effect.
  If $\parent{l}$ has a sibling $a_j$ and we cut $\edge{a_i}$, we can obtain
  an alternate MAF by cutting $\edge{l}$ instead of $\edge{a_i}$, resolving
  $\set{a_i, a_j}$, cutting edge $\edge{(a_i, a_j)}$ instead of $\edge{a_j}$
  if $\edge{a_j} \in E$, and otherwise cutting the same edges as in $E$.
  Thus, we can always replace the recursive call
  $\alg{F_1, F_2 \div \set{\edge{a_1}}, k-1, a_2}$ with the call
  $\alg{F_1, F_2 \div \set{\edge{a_1}, \edge{l}}, k-2, a_2}$ without affecting
  the correctness of the algorithm.
  If this call returns ``yes'', however, then so does the call
  $\alg{F_1, F_2 \div \set{\edge{B_1}, \edge{B_1'}},\break k-2, a_2}$ because we
  can yet again obtain an alternate MAF by cutting edges $\edge{B_1}$ and
  $\edge{B_2}$ instead of edges $\edge{a_1}$ and $\edge{l}$, resolving
  $\set{a_1, a_i}$, cutting $\edge{(a_1, a_i)}$ instead of $\edge{a_i}$ if
  $\edge{a_i} \in E$, and otherwise cutting the same edges as in $E$.
  Thus, the call $\alg{F_1, F_2 \div \set{\edge{a_1}}, k-1, a_2}$ can be
  eliminated altogether.
  An analogous argument shows that we can eliminate the call $\alg{F_1, F_2 \div
    \set{\edge{a_2}}, k-1, a_1}$.

  In Case~\ref{case:spr:mps-ladder}, Lemma~\ref{lem:spr:mps-take3} and the
  inductive hypothesis show that one of the recursive calls
  $\alg{F_1, F_2 \div \set{\edge{a_1}}, k-1, \nil}$,
  $\alg{F_1, F_2 \div \set{\edge{a_2}}, k-1, \nil}$,
  $\alg{F_1, F_2 \div \set{\edge{B_{11}}, \edge{B_{12}}, \dots,
      \edge{B_{1s_1}}}, k - s_1, \nil}$ or
  $\alg{F_1, F_2 \div \set{\edge{B_2}}, k - 1, \nil}$
  would return ``yes''.
  We need to show that the call $\alg{F_1, F_2 \div \set{\edge{a_2}}, k-1,
    \nil}$ is not necessary.
  To see this, observe that, if $l$ is a root, then cutting $\edge{a_2}$ or
  $\edge{B_2}$ has the same effect.
  If $l$ is not a root but has a member $a_i$ of the current sibling group as
  a sibling, then we can obtain an alternate MAF by cutting $\edge{B_2}$
  instead of $\edge{a_2}$, resolving $\set{a_2, a_i}$, cutting $\edge{(a_2,
    a_i)}$ instead of $\edge{a_i}$ if $a_i \in E$, and otherwise cutting the
  same edges as in $E$.

  In Case~\ref{case:spr:mps}, finally, the correctness follows from
  Lemmas~\ref{lem:spr:mps} and~\ref{lem:spr:mps-take3} if we can show that
  setting $a_0 \ne a_1$ is correct for the second and fourth recursive calls.
  This, however, follows because, if neither the first nor the
  third recursive call returns ``yes'', then in every MAF $F$ of $F_1$ and
  $F_2$, $a_1$ exists and there exist two leaves $a_1' \in F_2^{a_1}$ and
  $b_1' \in F_2^{B_{1j}}$, for some $1 \le j \le s_1$, such that
  $a_1' \reach[F] b_1'$.
\end{proof}

The proof of the following theorem is provided in
Section~\ref{sec:root-maf-appr} of the supplementary material.

\begin{theorem}
  \label{thm:approx:spr}
  Given two rooted $X$-trees $T_1$ and $T_2$, a \mbox{$3$-approximation} of
	$\ecut{T_1, T_2, T_2} = \drspr{T_1, T_2}$ can be computed in
	$\Oh{n \log n}$ time.
\end{theorem}


\section{Conclusions}

\label{sec:concl}

We developed efficient algorithms for computing MAFs of multifurcating trees.
Our fixed-parameter algorithm achieves
the same running time as in the binary case and our $3$-approximation
algorithm achieves a running time of $\Oh{n \log n}$, almost matching
the linear running time for the binary case.
Implementing and testing our algorithms will be the focus of future work.

Two other directions to be explored by future work are practical improvements
of the running time of the FPT algorithm presented here and extending
our FPT algorithm so it can be used to compute maximum \emph{acyclic} agreement
forests (MAAFs) and, hence, the hybridization number of multifurcating trees.
To speed up our FPT algorithm for computing MAFs, it may be possible to extend
the reduction rules used by Linz and Semple~\cite{linz09hnt} for computing MAAFs
of multifurcating trees so they can be applied to MAF computations, and combine
them with the FPT algorithm in this paper.
The fastest fixed-parameter algorithms for computing
MAAFs of binary trees~\cite{albrecht2012fast,chen2010hybridnet,
2011arXiv1108.2664W} are extensions of the binary MAF algorithms of
Whidden et al.~\cite{whidden2009uva,whidden2010fast}.
These algorithms were developed by examining which search branches of the binary
MAF algorithm get ``stuck'' with cyclic agreement forests and consider
cutting additional edges to avoid these cycles~\cite{albrecht2012fast} or
refine cyclic agreement forests to acyclic agremeent
forests~\cite{chen2010hybridnet,2011arXiv1108.2664W}.
Similarly, Chen and Wang~\cite{chen2012algorithms} recently
extended the MAF fixed-parameter algorithm for two binary trees
to compute agreement forests of multiple binary trees using
an iterative branching approach.
The proofs of various structural lemmas in this paper prove that \emph{any}
MAF can be obtained by cutting an edge set that includes certain edges.
To prove this, we started with an arbitrary MAF and an edge set such that
cutting these edges yields this MAF, and then we modified this edge set
so that it includes the desired set of edges without changing the resulting
forest.
As in the binary case \cite{2011arXiv1108.2664W}, the same lemmas apply
also to MAAFs; since the modifications in the proofs do not change the
resulting AF, the only needed change in the proof is to start with an
edge set such that cutting it yields an MAAF.
Thus, numerous lemmas in this paper may also form the basis for an efficient
algorithm for computing MAAFs.

Finally, we note that our fixed-parameter algorithm becomes greatly
simplified when comparing a binary tree to a multifurcating tree.
This is common in practice when, for example, comparing many multifurcating gene
trees to a binary reference tree or binary supertree.
To see this, suppose $F_1$ is binary, that is, $m=2$ in every case of the FPT
algorithm.
Then only Cases~\ref{case:spr:is}, \ref{case:spr:amops}~and~\ref{case:spr:mps2}
apply in Step~\ref{case:spr:non-sibling} and Step~\ref{case:tw:non-sibling}
never applies.
Using our observation that cutting $\edge{B_2}$ is not
necessary in Case~\ref{case:spr:amops} when $m=2$, our algorithm
becomes similar to the MAF algorithm for binary trees~\cite{whidden2010fast}.
We further note, in the interest of practical efficiency, that cutting $a_2$
is unnecessary in this case when the parent of $a_1$ is a binary node
(and, indeed, our algorithm is then identical to the algorithm of
\cite{whidden2010fast} when applied to two binary trees).


\section*{Acknowledgement}
Chris Whidden was supported by a Killam predoctoral scholarship and
the Tula Foundation, Robert G.\ Beiko is a Canada Research Chair and
was supported by NSERC, Genome Atlantic, and the Canada Foundation for
Innovation, and Norbert Zeh is a Canada Research Chair and was
supported by NSERC and the Canada Foundation for Innovation.


\nocite{clrs2}

\bibliographystyle{IEEEtran}
\bibliography{multi}

\renewcommand{\thesection}{S\arabic{section}}
\renewcommand{\thefigure}{S\arabic{figure}}

\section{Omitted Proofs}

\subsection{Lemma~\ref{lem:expand}}

\label{sec:lem:expand:proof}

The only difference between $F_2$ and $F_2'$ is the expansion of
$\set{a_{p+1}, a_{p+2}, \ldots, a_m}$ in $F_2'$, so
$\ecut{F_1, F_2, F_2} \le \ecut{F_1, F_2, F_2'}$.
Since $F \div E$ is an MAF of $F_1$ and $F_2$,
it suffices to show that $F \div E$ is an AF of $F_1$ and
$F_2'$ to prove that $\ecut{F_1, F_2, F_2'} \le \ecut{F_1, F_2, F_2}$.
Assume the contrary.
Then, since $F \div E$ is a forest of $F_1$, it cannot be a forest of~$F_2'$.
Since the only difference between $F_2$ and $F_2'$ is the
expansion of $\set{a_{p+1}, a_{p+2}, \ldots, a_m}$, this implies that some
component of $F \div E$ contains leaves $a_i' \in F_2^{a_i}$ and
$a_j' \in F_2^{a_j}$, for some $1 \le i \le p$ and $p+1 \le j \le m$,
contradicting that $a_i' \noreach [F \div E] a_j'$ for all such leaves.

\subsection{Lemma~\ref{lem:spr:twosiblings}}

\label{sec:lem:spr:twosiblings:proof}

As in the proof of Theorem~\ref{thm:spr:fourway},
(ii) follows immediately from (i), so
it suffices to prove that there exist a binary resolution
$F$ of $F_2$ and an edge set $E$ of size
$\ecut{T_1, T_2, F_2}$ such that $F \div E$ is an MAF of $T_1$ and
$T_2$ (and, hence, of $F_1$ and $F_2$)
and $E \cap \set{\edge{a_1}, \edge{a_2}, \edge{B_1}} \ne \emptyset$.
Once again, we show that, if $E \cap \set{\edge{a_1}, \edge{a_2}, \edge{B_1}}
= \emptyset$, we can replace an edge $f \in E$ with an edge in
$\set{\edge{a_1}, \edge{a_2}, \edge{B_1}}$
without changing $F \div E$.

This follows from the same arguments as in the proof of
Theorem~\ref{thm:spr:fourway}
unless there exist leaves $a_1' \in F_2^{a_1}$, $a_2' \in\nobreak
F_2^{a_2}$ and $b_1' \in F_2^{B_1}$ such that
$a_1' \reach[F \div E] \parent{a_1} \reach[F \div E] b_1'$ and
$a_2' \reach[F \div E] \parent{a_2}$.
In this case, since $a_1$ is not a child of $l$,
we have $a_2 \notin F_2^{B_1}$.
Thus, since $\set{a_1, a_2}$ is a sibling pair in $F_1$,
$F_1$ contains the triple $a_1'a_2'|b_1'$.
Since $F_2$ contains the triple $a_1'b_1'|a_2'$ and
$a_1' \reach[F \div E] b_1'$, this implies that $a_1' \noreach[F \div E] a_2'$.
Thus, we also have $a_2' \noreach[F \div E] x$, for all
$x \in F_2^{l} \setminus F_2^{a_2}$, as otherwise the components of
$F \div E$ containing $a_1', b_1'$ and~$a_2', x$ would overlap in $F_1$.
We choose an arbitrary leaf $b_2' \in B_2$ and
the first edge $f \in E$ on the path from $\parent{a_2}$ to~$b_2'$.
Lemma~\ref{lem:edge-shift} implies that $F \div E = F \div (E
\setminus \set{f} \cup \set{\edge{a_2}})$.

\subsection{Lemma~\ref{lem:spr:mps}}

\label{sec:lem:spr:mps:proof}

As in the proof of Lemma~\ref{lem:spr:mps2},
we prove this using induction on $s = s_1 + s_2$.
The base case is $s = 2$ and, hence, $s_1 = s_2 = 1$ because
$r > 2$ implies that $s_1 > 0$ and $s_2 > 0$.
In this case, Theorem~\ref{thm:spr:fourway} proves the lemma.
So assume $s > 2$ and the claim holds for all $1 \le s' < s$.
By Theorem~\ref{thm:spr:fourway}, there exist a resolution $F$ of
$F_2$ and an edge set $E$ of size $\ecut{T_1, T_2, F_2}$ such that $F
\div E$ is an AF of $F_1$ and~$F_2$ and
$E \cap \set{\edge{a_1}, \edge{a_2}, \edge{B_{11}},
  \edge{B_{21}}} \ne \emptyset$.
Using an inductive argument as in the proof of
Lemma~\ref{lem:spr:mps2}, it follows that there exist a resolution
$F'$ of $F_2$ and an edge set $E'$ satisfying the lemma.

\subsection{Lemma~\ref{lem:spr:mps-take3}}

\label{sec:lem:spr:mps-take3:proof}

First consider the case when $s_2 = 0$, which is possible because $r = 2$.
Then $B_2' = B_2$ and, by Theorem~\ref{thm:spr:fourway}, there exist a
resolution $F$ of $F_2$ and an edge set $E$ of size $\ecut{T_1, T_2, F_2}$
such that $F \div E$ is an AF of $T_1$ and $T_2$ and $E \cap
\set{\edge{a_1}, \edge{a_2}, \edge{B_{11}}, \edge{B_2}} \ne \emptyset$.
If  $E \cap \set{\edge{a_1}, \edge{a_2}, \edge{B_2}} \ne \emptyset$, the
lemma holds.
Otherwise $\edge{B_{11}} \in E$ and an inductive argument similar to the
one in the proof of Lemma~\ref{lem:spr:mps2} proves the lemma.

If $s_2 > 0$, we observe that the proof of Lemma~\ref{lem:spr:mps} did
not use the assumption that $r > 2$ but only that it implies $s_2 > 0$.
Hence, this proof shows that there exist a resolution $F$ of $F_2$
and an edge set $E$ of size $\ecut{T_1, T_2, F_2}$ such that $F \div E$
is an AF of $T_1$ and $T_2$ and $E \cap \set{\edge{a_1}, \edge{a_2}} \ne
\emptyset$, $\set{\edge{B_{11}, \edge{B_{12}}, \dots, \edge{B_{1s_1}}}}
\subseteq E$ or $\set{\edge{B_{21}}, \edge{B_{22}}, \dots, \edge{B_{2s_2}}}
\subseteq E$.
Among all such resolutions and edge sets, choose $F$ and $E$ so that
$E \cap \set{\edge{a_1}, \edge{a_2}} \ne \emptyset$ or $\set{\edge{B_{11}},
  \edge{B_{12}}, \dots, \edge{B_{1s_1}}} \subseteq E$ if possible.
If we can find such a pair $(F,E)$, the lemma holds.
Otherwise $\set{\edge{B_{21}}, \edge{B_{22}}, \dots, \edge{B_{2s_2}}}
\subseteq E$.
Let $F'$ be the forest obtained from $F_2$ by resolving $B_{21}, B_{22},
\dots, B_{2s_2}$ and cutting edges $\edge{B_{21}}, \edge{B_{22}}, \dots,
\edge{B_{2s_2}}$.
In $F'$, we have $r = 2$ and \mbox{$s_2 = 0$}.
Hence, by the argument in the previous paragraph, there exist a resolution
$F''$ of $F'$ and an edge set $E''$ of size $\ecut{T_1, T_2, F'}$ such
that $F'' \div E''$ is an AF of $T_1$ and $T_2$ and
$E'' \cap \set{\edge{a_1}, \edge{a_2}} \ne \emptyset$,
$\set{\edge{B_{11}}, \edge{B_{12}}, \ldots, \edge{B_{1s_1}}}
\subseteq E''$ or $\edge{B_2'} \in E''$.
In the first two cases, we obtain a contradiction to the choice of $F$ and
$E$.
In the latter case, the set
\mbox{$\set{\edge{B_{21}}, \edge{B_{22}}, \dots, \edge{B_{2s_2}}} \cup E''$}
has size $s_2 + \ecut{T_1, T_2, F''} = \ecut{T_1, T_2, F_2}$,
$F_2 \div (\set{\edge{B_{21}}, \edge{B_{22}}, \dots, \edge{B_{2s_2}}}
\cup E'')$ is an AF of $T_1$ and $T_2$, and
$\set{\edge{B_{21}}, \edge{B_{22}}, \dots, \edge{B_{2s_2}}, \edge{B_2'}}
\subseteq \set{\edge{B_{21}}, \edge{B_{22}}, \dots, \edge{B_{2s_2}}} \cup
E''$.
Thus, the lemma holds in this case as well.

\section{Linear Time Per Invocation}

\label{sec:linear-time}

We represent each forest as a collection of nodes, each of which points
to its parent, to its leftmost child, and to its left and right siblings.
This allows us to cut an edge in constant time, given the parent and
child connected by this edge.
Every labelled node (i.e., every node in $R_t$ or~$R_d$) stores a pointer to
its counterpart in the other forest.
For $\dot{T}_1$, we maintain a list of sibling groups of labelled nodes.
For each such group, the list stores a pointer to the parent of the sibling
group, which allows us to access the members of the sibling group by traversing
the list of the parent's children.
To detect the creation of such a sibling group, and add it to the list,
each internal node of $\dot{T}_1$ stores the number of its unlabelled children.
When labelling a non-root node, we decrease its parent's unlabelled children
count by one.
If this count is now~$0$, the children of this parent node form a new sibling
group, and we add a pointer to the parent to the list of sibling groups.
For $\dot{F}_2$, we maintain a list $R'_d \subseteq R_t$ of labelled nodes
that are roots of $\dot{F}_2$.
This list is used to move these roots from $R_t$ to $R_d$.

Steps~\ref{case:abort}--\ref{item:choose-sib-pair} are implemented
similarly to the algorithm for binary trees~\cite{whidden2010fast}.
Step~\ref{case:abort} clearly takes constant time.
In Step~\ref{case:success}, we can test in constant time whether
$\size{R_t} \le 1$ by inspecting at most two nodes in the first two
sibling groups.
Step~\ref{case:spr:singleton} takes constant time to test whether the
root list $R_d'$ is empty and, if it is not, cut the appropriate edge in
$\dot{T}_1$ and update a constant number of lists and pointers.
Step~\ref{item:choose-sib-pair} takes contant time using the list
of sibling groups.
We always choose the next sibling group from the beginning of this list
and append new sibling groups to the end.
This automatically gives preference to the most recently chosen
sibling group as required in Step~\ref{item:choose-sib-pair}.

Step~\ref{case:spr:sibling} requires some care to implement efficiently.
We iterate over the members $a_1, a_2, \dots, a_m$ of the current sibling
group and mark their parents in $\dot{F}_2$.
Initially, all nodes in $\dot{F}_2$ are unmarked.
When inspecting a node $a_i$ whose parent $p_{a_i}$ in $\dot{F}_2$ is
unmarked, we mark $p_{a_i}$ with~$a_i$.
If $p_{a_i}$ is already marked with a node $a_j$ ($j < i$), then $a_i$ and
$a_j$ are siblings in $\dot{T}_1$ and $\dot{F}_2$.
We resolve them in constant time and mark $p_{a_i}$ (which is now the
grandparent of~$a_i$) with the new parent $(a_i, a_j)$ of $a_i$ and~$a_j$.
Since we spend constant time per member $a_i$ of the sibling group, this
procedure takes $\Oh{m}$ time.
Once it finishes, the remaining members of the sibling group are not siblings
in $\dot{F}_2$.
Performing a contraction if the remaining sibling group has only one
member takes constant time.

In Step~\ref{item:choose-order}, we perform a linear-time traversal
of $\dot{F}_2$ to label every node $x$ with the number $r_x$ of members of the
current sibling group among its descendants.
Then, if the previously chosen minimal LCA $l$ still exists in $\dot{F}_2$
and has at least two descendants in the current sibling group, we keep this
choice of $l$.
Otherwise a node $x$ is a minimal LCA of a subset of the current sibling
group if and only if $r_x \ge 2$ and $r_y \le 1$, for each child $y$ of $x$.
If there is no such node $x$, we proceed to Step~\ref{case:tw:non-sibling}
without choosing $l$ because $a_i \noreach[F_2] a_j$, for all
$1 \le i < j \le m$.
Otherwise we pick any node $x$ satisfying this condition as the new minimal
LCA $l$.
No matter whether $l$ is the previously chosen minimal LCA or a new node,
we set $r = r_l$ and traverse the paths from $l$ to its descendant members
of the sibling group, $a_1, a_2, \dots, a_r$.
We do this by visiting all descendants $y$ of $l$ such that $r_y = 1$.
For all $1 \le i \le r$, the length of the path from $l$ to $a_i$, excluding
$l$ and $a_i$, is $s_i$.
We sort $a_1, a_2, \dots, a_r$ by their path lengths $s_1, s_2, \dots, s_r$
using Counting Sort~\cite{clrs2}.
Since $s_i \le n$, for all $1 \le i \le r$, this takes linear time.

To distinguish between Steps~\ref{case:tw:non-sibling}
and~\ref{case:spr:non-sibling}, it suffices to examine $a_0$.
We distinguish between the cases in Steps~\ref{case:tw:non-sibling}
and~\ref{case:spr:non-sibling} using the values of $r$, $m$, and
$s_1, s_2, \dots, s_r$ and, in Step~\ref{case:tw:non-sibling}, by
testing whether $a_0$ is among the descendants of $l$.
In each case, we can easily copy the forests,
cut the appropriate edges, and update our lists and pointers in
linear time for each of the recursive calls.

To summarize:
Each execution of Steps~\ref{case:abort}--\ref{item:choose-sib-pair} takes
constant time.
Step~\ref{case:abort} is executed once per invocation.
Steps~\ref{case:success}--\ref{item:choose-sib-pair} are executed
at most a linear number of times per invocation because each execution, except
the first one, is the
result of finding a root of $\dot{F}_2$ in Step~\ref{case:spr:singleton}
or resolving sibling pairs in Step~\ref{case:spr:sibling}, both of
which can happen only $\Oh{n}$ times.
Each execution of Step~\ref{case:spr:sibling} takes $\Oh{m}$ time.
In a given invocation, Step~\ref{case:spr:sibling} is executed at most
once per sibling group (because we either proceed to
Step~\ref{item:choose-order} or return to Step~\ref{case:success} after
completely resolving the sibling group).
Thus, since the total size of all sibling groups is bounded by
$\size{\dot{T}_1}$, the total cost of all executions of
Step~\ref{case:spr:sibling} per invocation is $\Oh{n}$.
Steps~\ref{item:choose-order}--\ref{case:spr:non-sibling} are executed
at most once per invocation and take linear time.
Thus, each invocation of the algorithm takes linear time.

\section{A 3-Approximation Algorithm for\\Rooted MAF}

\label{sec:root-maf-appr}

We now show how to modify the FPT algorithm from Section~\ref{sec:fpt} to obtain
a $3$-approximation algorithm with running time $\Oh{n \log n}$.
This algorithm is easy to implement iteratively, and this may be preferable in
practice.
In order to minimize the differences to the FPT algorithm, however, we
describe it as a recursive algorithm.
There are four differences to the FPT algorithm:
\begin{itemize}
\item Instead of deciding whether $\ecut{T_1, T_2, F_2} \le k$, an invocation
  $\alg{F_1, F_2}$ returns an integer $k''$ such that $\ecut{T_1, T_2, F_2} \le
  k'' \le 3 \ecut{T_1, T_2, F_2}$.
  Thus, there is no need for a parameter $k$ to the invocation or for
  an equivalent of Step~\ref{case:abort} of the FPT algorithm, and
  whenever Step~\ref{case:success} of the FPT algorithm would have returned
  ``yes'', we now return $0$ as our approximation $k''$ of
  $\ecut{T_1, T_2, F_2}$ because $F_2$ is an AF of $T_1$ and~$T_2$.
\item We execute Step~\ref{case:spr:sibling} only if the immediately preceding
  execution of Step~\ref{item:choose-sib-pair} chose a new sibling group.
  This ensures that this step is executed only once per sibling group.
  As discussed in Section~\ref{sec:linear-time},
  the cost per sibling group $\set{a_1, a_2, \dots, a_m}$ is $\Oh{m}$.
  Thus, the total cost of all executions of Step~\ref{case:spr:sibling}
  in all recursive invocations is $\Oh{n}$.
  After executing Step~\ref{case:spr:sibling}, implemented as discussed in
  Section~\ref{sec:linear-time}, every node $p$ in
  $F_2$ has at most one node $a_i$ as a child, which is stored as $p$'s
  \emph{representative}~$r_p$.
  This allows us to merge sibling pairs $\set{a_i, a_j}$ that arise as a result
  of edge cuts after we executed Step~\ref{case:spr:sibling} without
  re-executing this step:
  Whenever we contract a \mbox{degree-$2$} vertex whose only child is a node
  $a_i$ in the current sibling group and whose parent is $p$, we set
  $r_p := a_i$ if $r_p = \nil$; otherwise we resolve the sibling pair
  $\set{a_i, r_p}$ as in Step~\ref{case:spr:sibling} and store
  $(a_i, r_p)$ as $p$'s new representative.
\item We do not execute Step~\ref{item:choose-order}, as this would require
  linear time per invocation.
  Instead, we assign a \emph{depth estimate} $\dest{x}$ to each node $x$ and use
  it to choose the order in which to inspect the members of the current sibling
  group.
  Initially, $\dest{x}$ is $x$'s depth in $T_2$, which is easily computed
  in linear time for all nodes $x \in T_2$.
  In general, $\dest{x}$ is one more than the depth of $\parent{x}$ in~$T_2$,
  where $\parent{x}$ is $x$'s parent in $F_2$.
  In particular, $\dest{x}$ is an upper bound on $x$'s depth in $F_2$ and, for
  two nodes $x$ and $y$ with LCA $l$, we have $\dest{y} > \dest{x}$ if
  $x$ is a child of $l$ and $y$ is not.
  When choosing a new sibling group in Step~\ref{item:choose-sib-pair}, we
  insert all group members into a max-priority queue $Q$, with their depth
  estimates as their priorities.
  When contracting the degree-$2$ parent $\parent{x}$ of a node $x$, we set
  $\dest{x} := \dest{\parent{x}}$.
  If $x$ is a member of the current sibling group, we update its priority in
  $Q$.
  When resolving a sibling pair $\set{a_i, a_j}$, we remove $a_i$ and $a_j$
  from $Q$, set $\dest{(a_i, a_j)} := \dest{a_i}$, and insert $(a_i, a_j)$
  into $Q$.
  Finally, when cutting an edge~$\edge{a_i}$, for a member $a_i$ of the current
  sibling group, we remove $a_i$ from $Q$.
  These updates take $\Oh{\log n}$ time per modification of $F_2$.
  Since we modify $F_2$ at most $\Oh{n}$ times, the total cost of all priority
  queue operations is $\Oh{n \log n}$.
\item We do not distinguish between Steps~\ref{case:tw:non-sibling}
  and~\ref{case:spr:non-sibling} (having no concept of $a_0$) and also do not
  distinguish between the various cases of the steps.
  Instead, we have a single Step~\ref{case:tw:non-sibling} with four cases
  that each make one recursive call (see Figure~\ref{fig:spr:cases-approx}).
  \begin{enumerate}[label=7.\arabic{*}.,leftmargin=*,ref=7.\arabic{*}]
  \item \label{case:approx:m=2,new}
    If $m = 2$ and this sibling group was chosen in
    Step~\ref{item:choose-sib-pair} of the current invocation, make one
    recursive call $\alg{F_1, F_2 \div \set{\edge{a_1},\edge{\parent{a_1}},
        \edge{a_2}}}$ and return $3$ plus its return value.
  \item \label{case:approx:m=2,old}
    If $m = 2$ and this sibling group was chosen in
    Step~\ref{item:choose-sib-pair} of a previous invocation, make one
    recursive call $\alg{F_1, F_2 \div \set{\edge{a_1},\edge{\parent{a_1}},
        \edge{a_2}, \edge{\parent{a_2}}}}$ and return $4$ plus its return value.
  \item \label{case:approx:m>3,new}
    If $m > 2$ and this sibling group was chosen in
    Step~\ref{item:choose-sib-pair} of the current invocation, let $a_1$ and
    $a_2$ be the two entries with maximum priority in $Q$, ordered so that
    $\dest{a_1} \ge \dest{a_2}$.
    If $a_2$'s parent has a single sibling, this sibling is a member $a_j$
    of the current sibling group, and $a_1$'s parent is either a root or has
    a sibling that is not a member of the current sibling group, let $x = 2$;
    otherwise let $x = 1$.
    Remove $a_x$ from $Q$, make one recursive call $\alg{F_1, F_2 \div
      \set{\edge{a_x}, \edge{\parent{a_x}}}}$,
    and return $2$ plus its return value.
  \item \label{case:approx:m>3,old}
    If $m > 2$ and this sibling group was chosen in
    Step~\ref{item:choose-sib-pair} of a previous invocation,
    delete the node $a_1$ with maximum priority from $Q$,
    make one recursive call $\alg{F_1, F_2 \div \set{\edge{a_1},
        \edge{\parent{a_1}}}}$, and return $2$ plus its return value.
  \end{enumerate}
\end{itemize}
Using these modifications, we obtain the following theorem.

\begin{figure*}[t]
  \centering
  \includegraphics{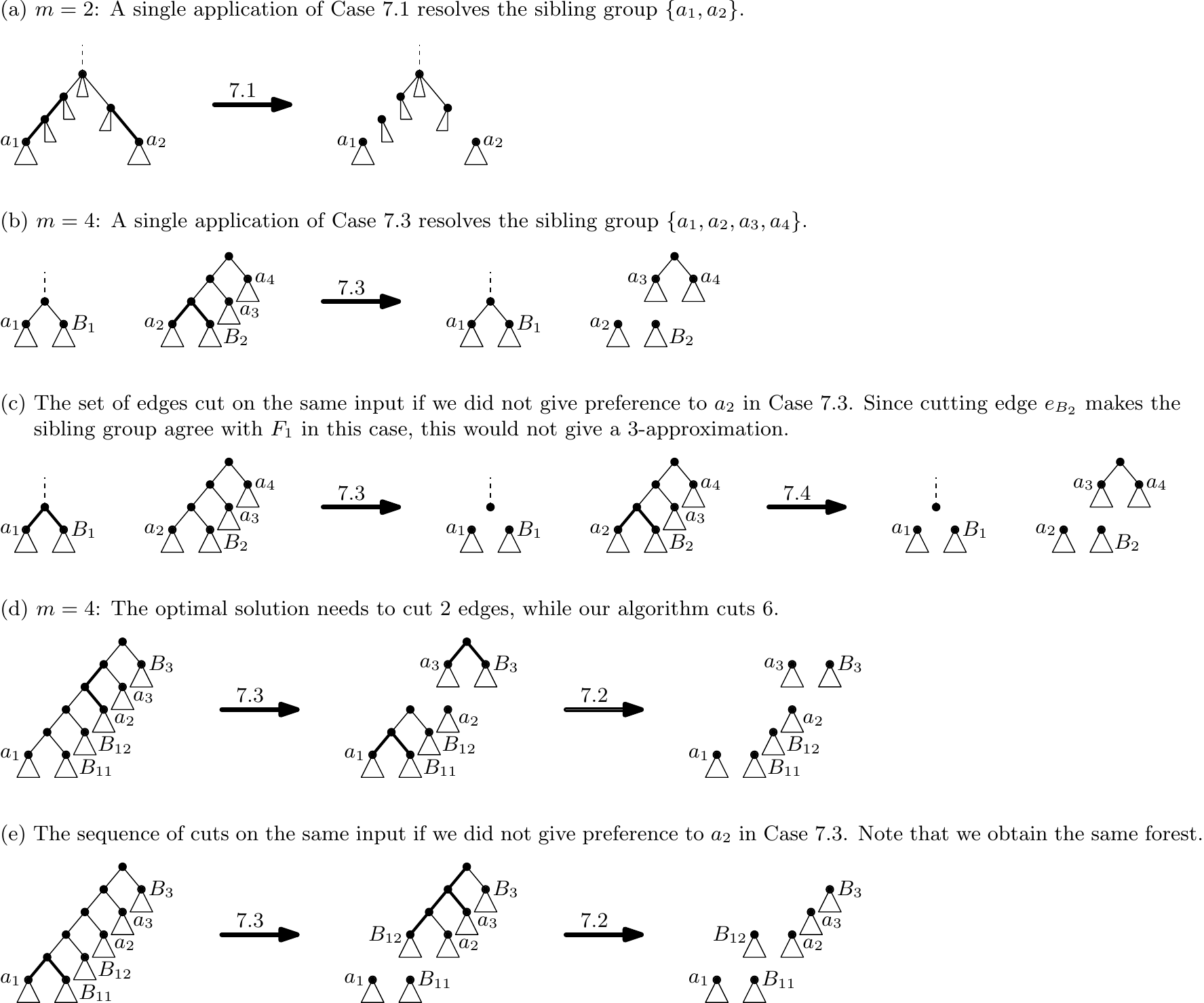}
  \caption{Illustration of the various cases of Step~\ref{case:tw:non-sibling}
    of the approximation algorithm.
    Only $\dot{F}_2$ is shown.
    Edges that are cut in each step are shown in bold.}
  \label{fig:spr:cases-approx}
\end{figure*}

\setcounter{theorem}{2}
\begin{theorem}
  Given two rooted $X$-trees $T_1$ and $T_2$, a \mbox{$3$-approximation} of
	$\ecut{T_1, T_2, T_2} = \drspr{T_1, T_2}$ can be computed in
	$\Oh{n \log n}$ time.
\end{theorem}

\begin{proof}
  We use the algorithm just described.
  This algorithm consists of Steps~\ref{case:success}--\ref{case:spr:sibling} of
  the FPT algorithm plus the modified Step~\ref{case:tw:non-sibling} above.
  In addition, there is a linear-time preprocessing step for computing the
  initial depth estimates of all nodes of $T_2$.
  We argued already that all executions of Step~\ref{case:spr:sibling} take
  linear time in total.
  In Section~\ref{sec:linear-time}, we argued that
  each execution of Step~\ref{case:success}, \ref{case:spr:singleton} or
  \ref{item:choose-sib-pair} of the FPT algorithm takes constant time.
  In the approximation algorithm, if Step~\ref{item:choose-sib-pair} chooses a
  new sibling group, it also needs to insert the members of the sibling group
  into the priority queue.
  This takes $\Oh{m \log m}$ time, $\Oh{n \log n}$ in total for all sibling
  groups.
  Each execution of Step~\ref{case:tw:non-sibling} takes $\Oh{\log n}$ time,
  constant time for the modifications of $F_2$ it performs and $\Oh{\log n}$
  time for the $\Oh{1}$ corresponding priority queue operations.
  Thus, to obtain the claimed time bound of $\Oh{n \log n}$ for the entire
  algorithm, it suffices to show that each step of the algorithm is executed
  $\Oh{n}$ times.

  This is easy to see for Steps~\ref{case:spr:singleton},
  \ref{case:spr:sibling}, and~\ref{case:tw:non-sibling}:
  Each execution of Step~\ref{case:spr:sibling} reduces the number of nodes
  in $R_t$ by one, the number of nodes in $R_t$ never increases, and initially
  $R_t$ contains the $n$ leaves of $T_1$.
  Each execution of Step~\ref{case:spr:singleton} or~\ref{case:tw:non-sibling}
  cuts at least one edge in $F_1$ or $F_2$, and initially these two forests
  have $\Oh{n}$ edges.

  For Steps~\ref{case:success} and~\ref{item:choose-sib-pair}, we
  observe that they cannot be executed more often than
  Steps~\ref{case:spr:singleton}, \ref{case:spr:sibling},
  and~\ref{case:tw:non-sibling} combined because any two executions of
  Step~\ref{case:success} or~\ref{item:choose-sib-pair} have an
  execution of Step~\ref{case:spr:singleton}, \ref{case:spr:sibling}
  or~\ref{case:tw:non-sibling} between them.

  It remains to bound the approximation ratio of the algorithm.
  First observe that the value $k'$ returned by the algorithm satisfies
  $k' \ge \ecut{T_1, T_2, T_2}$ because the input forest $F_2$ of the
  final invocation $\alg{F_1, F_2}$ is an AF of $T_1$ and $T_2$ and the
  algorithm returns the number of edges cut to obtain $F_2$ from $T_2$.
  To prove that $k' \le 3\ecut{T_1, T_2, T_2}$, let $r$ be the number
  of descendant invocations of the current invocation $\alg{F_1, F_2}$, not
  counting the current invocation itself, and let $k''$ be its return value.
  We use induction on $r$ to prove that $k'' \le 3\ecut{T_1, T_2, F_2}$ if
  $\alg{F_1, F_2}$ chooses a new sibling group in
  Step~\ref{item:choose-sib-pair}.
  We call such an invocation a \emph{master invocation}.
  We also consider the last invocation of the algorithm to be a master
  invocation.
  For two master invocations without another master invocation between them,
  all invocations between the two invocations are \emph{slave invocations}
  of the first of the two master invocations, as they manipulate the sibling
  group chosen in this invocation.

  As a base case observe that, if $r = 0$, then $\ecut{T_1, T_2, F_2} = 0$ and
  the invocation $\alg{F_1, F_2}$ returns $k'' = 0$ in Step~\ref{case:success}.
  For the inductive step, consider a master invocation $\I = \alg{F_1, F_2}$
  with $r > 0$, and let $\I' = \alg{F_1', F_2'}$ be the first master invocation
  after~$\I$.
  By the inductive hypothesis, $\I'$ returns a value $k'''$ such that
  $k''' \le 3\ecut{T_1, T_2, F_2'}$.

  If $m = 2$ in invocation $\I$, we invoke Case~\ref{case:approx:m=2,new},
  which cuts edges $\edge{a_1}$, $\edge{\parent{a_1}}$, and $\edge{a_2}$ and,
  thus, makes the sibling group agree between $F_1$ and $F_2$ (see
  Figure~\ref{fig:spr:cases-approx}(a)).
  This implies that $k'' = k''' + 3$.
  By Lemma~\ref{lem:spr:twosiblings}, we have $\ecut{T_1, T_2, F_2'} \le
  \ecut{T_1, T_2, F_2} - 1$.
  Since $k''' \le 3\ecut{T_1, T_2, F_2'}$, this shows that
  $k'' \le 3\ecut{T_1, T_2, F_2}$.

  If $m > 2$ in invocation $\I$, this invocation applies
  Case~\ref{case:approx:m>3,new}, each of its slaves applies
  Case~\ref{case:approx:m=2,old} or~\ref{case:approx:m>3,old}, and
  Case~\ref{case:approx:m=2,old} is applied at most once.
  Since the sibling group $\set{a_1, a_2, \dots, a_m}$ does not agree between
  $F_1$ and~$F_2$, at least one edge cut in $F_2$ is necessary to make the
  sibling group agree between $F_1$ and $F_2$.
  We distinguish whether one or more edge cuts are required.

  If one cut suffices (Figures~\ref{fig:spr:cases-approx}(b)
  and~\ref{fig:spr:cases-approx}(c)), then
  there are at most two components of $F_2$ that
  contain members of the current sibling group because resolving overlaps
  in $F_1$ between $q$ components of $F_2$ requires at least $q - 1$ cuts
  in $F_2$.
  Consequently, exactly one component $C$ contains at least two members of
  the sibling group.
  The existence of at least one such component follows because $m > 2$.
  If we had another component $C'$ containing at least two members of the
  current sibling group, then at least one cut would be required in each of $C$
  and $C'$ to make the sibling group agree between $F_1$ and $F_2$, but we
  assumed that one cut suffices.

  For a single cut to suffice to make the current sibling group agree between
  $F_1$ and $F_2$, $C$ must consist of a single path of nodes
  $x_1, x_2, \dots, x_t$ such that, for $1 \le j < t$, $x_j$ has two children:
  $x_{j+1}$ and a member $a_{i_j}$ of the current sibling group; $x_t$ has a
  member $a_{i_t}$ of the current sibling group as a child, as well as a group
  $B_{i_t}$ of siblings of $a_{i_t}$ such that no member $a_h$ of the current
  sibling group belongs to~$F_2^{B_{i_t}}$.
  Thus, cutting edges $\edge{a_{i_t}}$ and $\edge{\parent{a_{i_t}}}$ makes the
  sibling group agree between $F_1$ and $F_2$.
  Now observe that $a_{i_t}$ and $a_{i_{t-1}}$ are the two members of the
  current sibling group with the greatest depth estimates in $C$.
  If there exists another component $C'$ of $F_2$ that contains a member $a_h$
  of the current sibling group, then $a_h$ is the only such node in $C'$.
  Thus, the two maximum priority entries in $Q$ are either $a_{i_t}$ and
  $a_{i_{t-1}}$ or $a_{i_t}$ and $a_h$.
  In both cases, invocation $\I$ cuts edges $\edge{a_{i_t}}$ and
  $\edge{\parent{a_{i_t}}}$ because $a_h$ either has no parent or its parent
  does not have a member of the current sibling group as a sibling,
  and $a_{i_t}$ is preferred over $a_{i_{t-1}}$ by invocation $\I$ because
  $a_{i_t}$'s parent does have a member of the current sibling group as its only
  sibling (namely $a_{i_{t-1}}$) and has a greater depth estimate than
  $a_{i_{t-1}}$.
  Thus, in $\I$'s child invocation, the current sibling group agrees between
  $F_1$ and $F_2$, which implies that this child invocation is $\I'$ and
  $k'' = k''' + 2$.
  Moreover, since the sibling group does not agree between $F_1$ and $F_2$
  in invocation $\I$, we must have $\ecut{T_1, T_2, F_2'} \le
  \ecut{T_1, T_2, F_2} - 1$ and, hence,
  $k'' = k''' + 2 \le 3\ecut{T_1, T_2, F_2'} + 2 \le 3\ecut{T_1, T_2, F_2}$.

  For the remainder of the proof assume at least two cuts are necessary to make
  the sibling group $\set{a_1, a_2, \dots, a_m}$ agree between $F_1$ and $F_2$
  (Figures~\ref{fig:spr:cases-approx}(d) and~\ref{fig:spr:cases-approx}(e)), and
  assume the members of the sibling group are ordered by their depth estimates.
  Let $i_1, i_2, \dots, i_s$ be the indices such that invocation $\I$ and its
  slaves cut edges $\set{\edge{a_{i_j}}, \edge{\parent{a_{i_j}}} \mid
    1 \le j \le s}$ and, for all $1 \le j \le s$, let $B_{i_j}$ be the set of
  $a_{i_j}$'s siblings at the time we cut edges $\edge{a_{i_j}}$
  and~$\edge{\parent{a_{i_j}}}$.
  Assume for now that invocation $\I$ cuts edges $\edge{a_1}$
  and~$\edge{\parent{a_1}}$.
  Then, for all $1 \le j \le s$, $F_2^{B_{i_j}}$ contains no member of the
  current sibling group because such a member $a_h$ would have a greater depth
  estimate than $a_{i_j}$ and hence $\edge{a_h}$ would have been cut before
  $\edge{a_{i_j}}$.
  This implies that there exists an edge set $E$ such that $F_2 \div E$ is an
  MAF of $F_1$ and $F_2$ and $\size{E \cap \set{\edge{a_{i_j}}, \edge{B_{i_j}}
      \mid 1 \le j \le s}} \ge s - 1$.
  Since cutting edges $\edge{a_{i_j}}$ and $\edge{B_{i_j}}$ produces the same
  result as cutting edges $\edge{a_{i_j}}$ and $\edge{\parent{a_{i_j}}}$, this
  shows that $\ecut{T_1, T_2, F_2'} \le \ecut{T_1, T_2, F_2} - (s - 1)$.

  If $s \ge 3$, we have $2s \le 3(s - 1)$ and, hence, $k'' \le k''' + 3(s - 1)
  \le 3(\ecut{T_1, T_2, F_2'} + s - 1) \le 3\ecut{T_1, T_2, F_2}$.
  If $s < 3$, we observe that
  $\ecut{T_1, T_2, F_2'} \le \ecut{T_1, T_2, F_2} - 2$ because the current
  sibling group agrees between $F_1$ and $F_2'$ and we assumed that at least two
  cuts are necessary in $F_2$ to make this sibling group agree between $F_1$ and
  $F_2$.
  Thus, $k'' \le k''' + 2s \le 3\ecut{T_1, T_2, F_2'} + 4 \le
  3\ecut{T_1, T_2, F_2}$.

  It remains to deal with the case when invocation $\I$ cuts edges $\edge{a_2}$
  and $\edge{\parent{a_2}}$.
  If $a_1 \notin F_2^{B_2}$, then again $F_2^{B_{i_j}}$ contains no member of
  the current sibling group, for all \mbox{$1 \le j \le s$}, and the same
  argument as above shows that $k'' \le 3\ecut{T_1, T_2, F_2}$.
  If $a_1 \in F_2^{B_2}$, then observe that either the path in $F_2$ between
  $a_1$ and $\parent{a_2}$ has at least two internal nodes or $a_2$ has at least
  two siblings because otherwise invocation $\I$ would prefer $a_1$ over $a_2$
  because $a_1$ has the greater depth estimate.
  This implies that $a_2$ has the same parent before and after cutting edges
  $\edge{a_1}$ and $\edge{\parent{a_1}}$, and $a_1$ has the same parent before
  and after cutting edges $\edge{a_2}$ and $\edge{\parent{a_2}}$.
  Thus, $\I$ and its child invocation cut the same four edges
  $\edge{a_1}$, $\edge{\parent{a_1}}$, $\edge{a_2}$, and $\edge{\parent{a_2}}$
  as would have been cut if invocation $\I$ had not preferred $a_2$ over $a_1$.
  The same argument as above now shows that $k'' \le 3\ecut{T_1, T_2, F_2}$.
\end{proof}




\end{document}